\documentclass[12pt.,onecolumn,draftcls]{IEEEtran}
\usepackage{amssymb}
\usepackage{amsthm,bm}
\usepackage{amsmath}
\usepackage{cite}
\usepackage{mathtools}
\usepackage{array}
\newcolumntype{L}{>{\centering\arraybackslash}m{5cm}}
\usepackage{multicol}
\usepackage{multirow}
\usepackage{booktabs}
\usepackage[utf8]{inputenc}
\usepackage[english]{babel}	
\usepackage{subfigure}
\usepackage{caption}
\usepackage{minipage-marginpar}
\usepackage{xcolor}
\usepackage{enumerate}
\usepackage{color}
\usepackage[lined,boxed,ruled,vlined]{algorithm2e}

\newlength\myindent
\newtheorem{proposition}{Proposition}[section]

\ifCLASSINFOpdf
\else
  \usepackage[dvips]{graphicx}
  \graphicspath{{../eps/}}
\fi

\begin{document}

\title{Buffer-aided Resource Allocation for a Price Based Opportunistic Cognitive Radio Network}

\author{Nilanjan~Biswas, Goutam~Das, and Priyadip~Ray\\
G.S.Sanyal School of Telecommunications, IIT Kharagpur, India}

\maketitle

\begin{abstract}
In this paper, a resource allocation problem for an opportunistic cooperative cognitive radio network is considered, where cognitive radio nodes send their hard decisions to the fusion center. The fusion center plays dual role, i.e., takes the global decision (i.e., decision about the primary user's activity) as well as allocates transmission time durations among cognitive radio nodes. Revenue based utility functions are considered at the fusion center and cognitive radio nodes. An optimization problem is formulated to maximize the fusion center's revenue while satisfying some well defined constraints. User selection among cognitive radio nodes is performed in order to make the optimization problem feasible.
\end{abstract}


\section{Introduction}
In last few years, the communication industry has witnessed massive growth in number of wireless devices as well as traffic. As per the report by Cisco \cite{cisco2016global}, number of global wireless users will rise to 5.5 billion in 2020 from 4.8 billion in 2015. The concept of spectrum sharing has emerged to cope with this ever increasing wireless data demand. In literature \cite{Nokia,malladi2016best,zhang2015lte,abdelraheem2018bargaining,mehrnoush2018fairness,mueck2015spectrum}, we find the usage of spectrum sharing concept in various forms like long term evolution-unlicensed (LTE-U), LTE-licensed assisted access (LTE-LAA), MULTEFIRE, licensed shared access (LSA), and spectrum access system (SAS). Spectrum sharing is considered one of the major requirements in 5G in order to make devices backward compatible \cite{WinNT,WinNT1}.

Cognitive radio (CR) \cite{haykin2005cognitive} is considered as an efficient spectrum sharing technology, whose application can be found in different modern wireless communications contexts like machine-to-machine (M2M) \cite{lee2013feasibility}, internet-of-things (IoT) \cite{khan2017cognitive}, network virtualization \cite{liang2015wireless}, and hotspots \cite{kim2012admission}. For more information about past standardizations and current initiatives on CR technology in spectrum sharing, one may refer to \cite{sherman2008ieee} and \cite{bhattarai2016overview} respectively. In CR networks, unlicensed secondary users (SUs) try to opportunistically access the primary user's (PU's) licensed spectrum with the help of spectrum sensing. In \cite{yucek2009survey}, we find different spectrum sensing techniques related to individual and cooperative sensing. Cooperative spectrum sensing (CSS) is regarded as one of the efficient sensing techniques as it is robust against wireless channel impairments \cite{ganesan2005cooperative,mishra2006cooperative,unnikrishnan2008cooperative}. In CSS, multiple SUs sense the PU's spectrum and send their sensing information to a central unit, which is commonly known as the fusion center (FC). Sensing information may be of two types, which are hard sensing information (i.e., single bit) and soft sensing information (i.e., multi bits). Hard sensing information is often preferred for it's bandwidth and energy efficiencies. The FC fuses all sensing information received from SUs and takes the global decision about the PU's activity.

Based on the outcome of spectrum sensing, the FC decides whether to allocate available resources among SUs and allow SUs to use the licensed spectrum or not. In the context of resource allocation in opportunistic CR networks, we observe that people have considered either multiple frequency resources \cite{fan2010optimal,fan2011joint,almalfouh2010uplink,xie2012dynamic,lim2012joint,wang2013resource,shahini2018joint,ejaz2018multi,kulkarni2017multi} or single frequency resource \cite{biswas2018optimal} for the primary network. In case of multiple frequency resources, people have considered frequency resource allocation \cite{fan2010optimal,shahini2018joint}, frequency resource and power allocation \cite{fan2011joint,almalfouh2010uplink,xie2012dynamic,lim2012joint,wang2013resource}, and time and energy allocation \cite{kulkarni2017multi}. Authors of \cite{biswas2018optimal} have considered time allocation for a single frequency resource based scenario. The FC performs the resource allocation in order to maximize a global utility function, e.g., secondary network's throughput \cite{almalfouh2010uplink,fan2011joint,xie2012dynamic,wang2013resource,kulkarni2017multi,shahini2018joint}, power consumption \cite{ejaz2018multi}, and economy based utility function \cite{lim2012joint,biswas2018optimal} have been considered in literature. In this paper, we consider single frequency resource for a CSS based CR network, where the FC performs time allocation among SUs to maximize an economy based utility function.

For a CSS based CR networks, it is important to design decision thresholds for SUs and the FC. Related literature includes \cite{biswas2018optimal}, where authors have jointly designed decision thresholds of SUs and the FC and have shown the efficacy of the joint design. Another two parameters which are relevant for resource allocation in a CSS based CR networks are set of SUs among which resource to be allocated \footnote{The need of SU selection may arise for a scenario where the available resource falls shorter than the required.} and amount of resource to be allocated to a selected SU. We observe that resource allocation procedure strongly depends on the detection performance of CR networks. As the detection performance depends on the number of SUs participating in the resource allocation procedure, it is important to design decision thresholds and other resource allocation parameters in a joint way. In literature, we find \cite{fan2011joint,lim2012joint}, where authors have considered joint design of decision thresholds and resource allocation. Authors of \cite{fan2011joint} have considered soft sensing fusion at the fusion center, such that, design of decision thresholds at the FC is important. On the other hand, decision thresholds for multiple bands have been evaluated in \cite{lim2012joint} for a non-CSS scenario. It is to be noted that in both of \cite{fan2011joint,lim2012joint}, authors have assumed sufficient available resources in order to fulfill SUs' demands. Therefore, SU selection problem is not relevant in \cite{fan2011joint,lim2012joint}. Unlike \cite{fan2011joint,lim2012joint}, in this paper, we consider CSS with hard decision fusion and we consider constraint on the available resource, which brings SU selection problem in our paper.

More specifically, we consider economic utilities at the FC and SUs. SUs pay to the FC to get access of the licensed spectrum, which form the FC's utility. Whereas, a SU's utility is constructed with two different economic cost factors, i.e., economic gain due to licensed spectrum usage and corresponding economic loss faced due to energy loss during spectrum sensing. We formulate an optimization problem to maximize the FC's utility while considering different constraints: (i) interference constraint, (ii) positivity constraints on SUs' utilities, (iii) upper bounds on SUs' time allocations from traffic sizes, and (iv) upper bound on total time allocation from the frame duration. The third constraint signifies that the FC does not allocate more time durations to a SU than what is required to clear the SU's buffer. Due to the upper bound on total time allocation, SU selection problem arises in this paper. Only selected SUs take part in sensing and further allocation process. Increasing the number of SUs participating in the detection process may help in reducing interference on the primary network from the secondary network in the cost of lower time duration for resource allocation. This motivates us in jointly designing different parameters, i.e., decision thresholds (of SUs' and the FC's), SU set selection, and time allocation to selected SUs. To the best of our knowledge, such kind of joint design has not been considered earlier in CR networks context.   

We next summarize our contributions as follows:

\begin{itemize}
 \item We design decision thresholds (for SUs and the FC), select set of SUs for resource allocation, and allocation time durations among selected SUs in a joint manner. 

 \item We consider heterogeneous data requirements for SUs. The FC has knowledge of these requirements and allocates resource accordingly. 

 \item The optimization problem consists of both integer and continuous variables which makes the solution process non-trivial. We analyse the optimization problem and devise an algorithm to find out solution, which we prove to be optimal under a particular scenario.
\end{itemize}

The remainder of the paper is organized as follows: Section~\ref{sysmod} presents the system model. In Section~\ref{ORA}, we discuss about different utility functions and constraints. Section~\ref{optimResource} represents the optimization problem which we consider in this paper. Corresponding algorithm for solving the optimization problem is given in Section~\ref{optimsolution}. In Section~\ref{result}, we show different results and describe the efficacy of our proposed algorithm. We conclude in Section~\ref{conclusion}.

\section{System Model}\label{sysmod}
In Fig.~\ref{cooperative}, we provide an overview of our system model, where a secondary network co-exist with a primary network. The secondary network opportunistically uses the primary network's licensed spectrum. We consider single primary transmitter (PT) and primary receiver (PR) in the primary network. In secondary network, there are multiple SUs, i.e., $SU_i$, $i=1,..,M$, and a FC. Corresponding set of all SUs is defined by $\mathcal{G}=\left\{SU_i\right\}_{i=1}^M$. SUs participate with the FC in CSS to detect the PT's activity. Once the PT is detected idle, SUs use the licensed spectrum to communicate with the FC. In Fig.~\ref{cooperative}, solid lines indicate intended signals, whereas, red dashed lines are for interfering signals and violet dashed lines are for sensing signals. For missed-detection, the FC and the PR receives interference from the PT and SUs respectively.
\begin{figure}[t!]
	\centering
	\includegraphics[trim=0cm 3cm 0cm 3cm,clip=true,width=12cm]{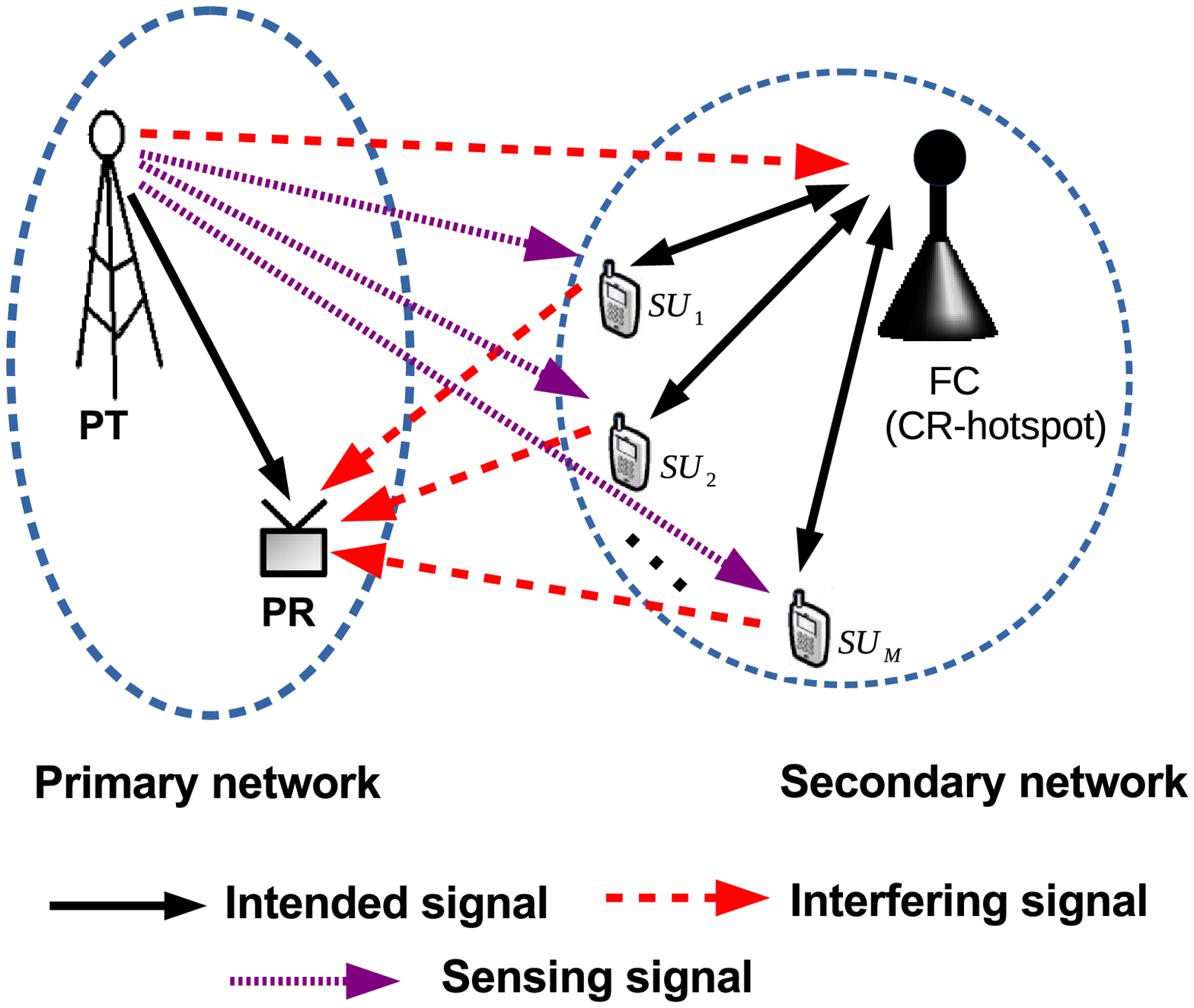}
	\DeclareGraphicsExtensions{.eps}
	\caption{System Model}
	\label{cooperative}
\end{figure}

\textit{Application of the System Model: }Now a days, we find significant usage of mobile hotspots at public places like airports, coffee shops, hotels etc. Hotspots may be created by using the unlicensed industrial, scientific, and medical (ISM) spectrum or licensed spectrum \cite{hotspot1}. We may map our system model with a CR enabled hotspot scenario, where the FC may behave as hotspot sharing the licensed spectrum of the primary network. However, the FC can share the licensed spectrum in an opportunistic way. As the sharing takes place in licensed spectrum, a business perspective arises for such CR enabled hotspot scenario. In Section~\ref{optimResource}, we discuss about different economic costs which are considered in our system model.

The secondary network operates in a time slotted way, where we denote a slot by a time frame. There are multiple phases in a frame duration, which we show in Fig.~\ref{frame} and then describe different phases in detail.
\begin{figure}[t!]
	\centering
	\includegraphics[trim=0cm 9cm 0cm 6cm,clip=true,width=12cm]{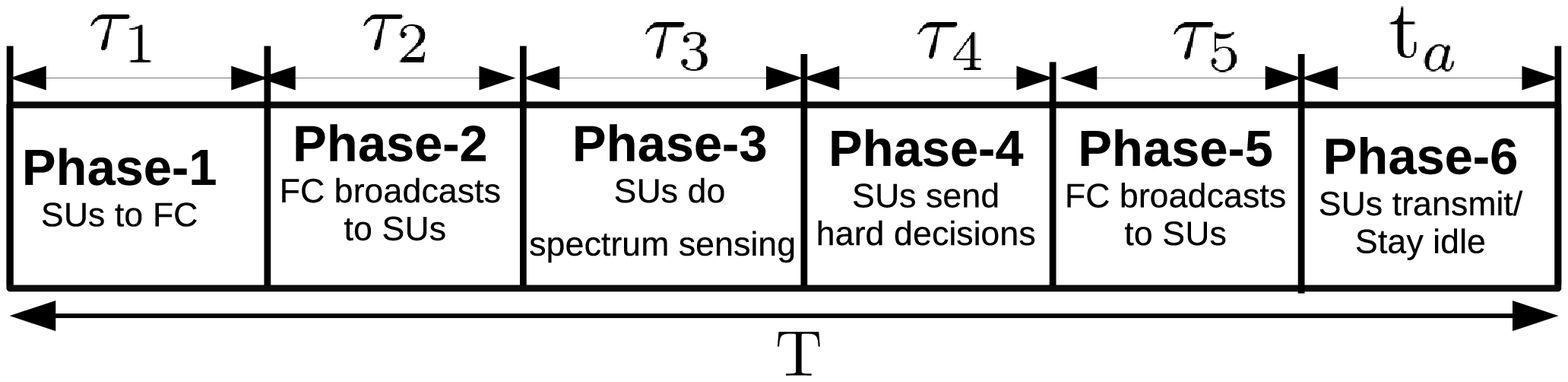}
	\DeclareGraphicsExtensions{.eps}
	\caption{Frame structure}
	\label{frame}
\end{figure}
\begin{enumerate}
	\item \textit{Phase-1:} At the beginning of a frame, each SU informs the FC about it's stored data at buffer and traffic type. SUs define their payments to the FC based on their traffic types. Communications between SUs and the FC take place over a common control channel which is used by SUs in a TDMA fashion. Required time is denoted by $\tau_1=M\tau_r$, where $\tau_r$ is the required reporting time by a SU.
	\item \textit{Phase-2:} The FC determines activity factors for SUs and broadcasts it to the secondary network. We define $\alpha_i\in\{0,1\}$, as the activity factor for $SU_i$, where $\alpha_i=1$ and 0 indicates that the $SU_i$ is allocated and not allocated time duration respectively for accessing the licensed spectrum. We denote the required time for broadcasting by $\tau_2$. Broadcasting is done on the control channel.
	\item \textit{Phase-3:} Active SUs perform spectrum sensing synchronously for $\tau_3=N\tau_s$, where $N$ is number of sensing samples and $\tau_s$ is the sampling interval.
	\item \textit{Phase-4:} After sensing, SUs send their hard decisions over the control channel in TDMA fashion. Total reporting time is $\tau_4=\sum_{i=1}^M \alpha_i\tau_r^{'}$, where $\tau_r^{'}$ is the reporting time of each SU.
	\item \textit{Phase-5:} The FC takes the final decision about the PT's activity with the help of collected hard decisions. If the PT is detected as idle, then the FC broadcasts the allotted time durations to active SUs over the control channel. Required broadcasting time is $\tau_5$.
	\item \textit{Phase-6:} The rest of the duration over which the SUs may access the licensed band becomes $t_a=T-\sum_{k=1}^5 \tau_k$. As SUs access the licensed channel orthogonally in time domain, we can write $\sum \alpha_i t_i \leq t_a$, where $t_i$ is the allotted time for data transmission for $SU_i$.
\end{enumerate}

We consider following assumptions in our system model:
\begin{itemize}
	\item Both primary and secondary networks follow synchronized frame structure. Communications between the SUs and the FC take place over a predefined control channel.
	\item During a frame duration, a SU possesses a particular traffic type. SUs offer prices to the FC based on their corresponding traffic types. 
	\item Channels between SUs and the FC are independent but not identical. The FC has knowledge about instantaneous channel gains.
	\item The FC has knowledge about the statistic of the channel, i.e., between the PT and the FC.
	\item Sensing channels, i.e., channels between the PT and SUs are independent and identical.  
\end{itemize}

\subsection{Local sensing procedure}
We denote the observation vector at $SU_i$ for spectrum sensing as $\bm{x_i}=[x_{i}(1),...,x_{i}(N)]$. We define the presence and absence of the PT by hypotheses $H_1$ and $H_0$ respectively, and write the received signal at $SU_i$ during $j\textsuperscript{th}$ sampling instant as:

\begin{equation}
 \begin{array}{l}
  H_1:\hspace{5 mm} x_{i}(j)=h_ip(j)+w_{i}(j)\\
  H_0:\hspace{5 mm} x_{i}(j)=w_{i}(j)\\ 
 \end{array}
\end{equation}
\noindent
where, $j=1,....,N$. $h_i$ is the channel coefficient between the PT and $SU_i$, which remains constant during a frame duration. $p(j)$ and $w_{i}(j)$ are the PT's signal during $j\textsuperscript{th}$ sampling instant and the additive white Gaussian noise (AWGN) at $SU_i$. We denote noise and the PT signal's variances by $\sigma_w^2$ and $\sigma_p^2$ respectively.

We choose energy detection at SUs as it is less complex \cite{yucek2009survey}. Test statistic for energy detection at $U_i$ for taking a hard decision (i.e., 0 or 1) is:
\begin{equation}
 y_i=\frac{1}{N}\sum_{j=1}^{N}|x_{i}(j)|^2\mathop{\gtreqless}_{H_0}^{H_1} \epsilon_i
\label{test_CR}
\end{equation}
\noindent
where, $\epsilon_i$ is the detection threshold at $SU_i$.

Under the assumption of sufficiently large value of $N$, false-alarm and detection probabilities of $SU_i$ may be written as\cite{liang2008sensing}:
\begin{subequations}
 \begin{align}
  P_{fa}^i &= Q\left\{\left(\frac{\epsilon_i}{\sigma_w^2}-1\right)\sqrt{N}\right\} \label{pf}\\	
  P_{d}^i  &= Q\left\{\frac{1}{\sqrt{2\gamma_i+1}}\left(Q^{-1}(P_{fa}^i)-\sqrt{N}\gamma_i\right)\right\}\label{pd}
 \end{align}
\end{subequations}\noindent
where, $Q(x)=\frac{1}{\sqrt{2\pi}}\int_x^\infty\exp\left(-\frac{t^2}{2}\right)dt$ and $\gamma_i=\frac{\left|h_i\right|^2\sigma_p^2}{\sigma_w^2}$ is the received SNR at $SU_i$ from the PT during sensing. 

We assume that SUs get identical sensing SNR, i.e., $\gamma_i=\gamma,\forall{SU_i\in \mathcal{G}}$, for which we get identical false-alarm and detection probabilities, i.e., $P_{fa}^i=P_{fa}$ and $P_d^i=P_d$. Please note that identical sensing SNR may take place when SUs are in much lesser distance among each other compared to distance between the PT and SUs. 

\subsection{Fusion process}
The FC collects active SUs' hard decisions (i.e., 0 or 1) and fuses them by using optimal Chair-Varshney fusion rule \cite{varshney1996distributed}. We can represent the FC's false-alarm and detection probabilities as functions of each SU's false-alarm probability, decision threshold, and number of hard decisions:
\begin{subequations}
 \begin{align}
  P_{FA}(P_{fa},k,L) &= \sum_{l=k}^{L} {L \choose l} P_{fa}^l(1-P_{fa})^{L-l} \label{Pf}\\	
  P_{D}(P_{fa},k,L)  &= \sum_{l=k}^{L} {L \choose l} P_{d}^l(1-P_{d})^{L-l} \label{Pd}
 \end{align}
\end{subequations}\noindent
where, $L=\sum_{SU_i\in \mathcal{G}} \alpha_i$, represents the cardinality of active SUs' set and $k$ is the decision threshold at the FC.

\section{Utility functions and constraints}\label{ORA}
In this section, we briefly describe utilities at the FC and SUs, and constraints under which the FC's utility is maximized during the resource allocation.    

\subsection{Different associated prices}\label{prices}
The resource allocation procedure strongly depends on different associated prices. In Table~\ref{my-label14}, we show different associated economic prices for different conditions (i.e., the PT's actual state and sensing result), when $SU_i$ is allocated time duration $t_i$. In third column of Table~\ref{my-label14}, we show that $SU_i$ pays amount $a_i$ to the FC for each successfully transmitted bit for two different detection scenarios, i.e., proper detection of the PT's absence and missed-detection of the PT's presence. For the former detection case, normalized rate at the FC from $SU_i$ becomes $r_0^i$, whereas, for the latter case, it becomes $r_1^i$, where, 
\begin{subequations}
	\begin{align}
	r_0^i &= B_{w}\log_2\left(1+\frac{g_{SU_i-FC}P_{ST}}{N_0}\right) \text{ bits/sec.}\label{r0}\\
	r_1^i &= B_{w}E_{g_{PT-FC}}\left[\log_2\left(1+\frac{g_{SU_i-FC}P_{ST}}{g_{PT-FC}P_{PT}+N_0}\right)\right] \text{ bits/sec.}\label{r1}
	\end{align}
\end{subequations}\noindent
where $B_w$ is the bandwidth of the channel of interest, $g_{SU_i-FC}$ is instantaneous channel power gain from $SU_i$ to the FC, $g_{PT-FC}$ is the channel power gain between the PT and the FC over which averaging is done ($E_{x}(.)$ is the symbol of expectation over the random variable $x$), $P_{ST}$ and $P_{PT}$ are SUs' and the PT's transmission powers respectively \footnote{We assume identical transmission powers for SUs.}, and $N_0$ is the received noise power at the FC. $SU_i$ generates revenue $b_i$ for each successful bit transmission, which is shown in fourth column of Table~\ref{my-label14}. SUs deplete energies while sensing and transmitting decisions, which we correspond to price $\mathcal{T}=Na_s+a_t$, as given in last column of Table~\ref{my-label14}, where $a_s$ is the cost for collecting a sample for sensing and $a_t$ is the cost for transmitting a hard decision. It is to be noted that different economic costs will be associated with $SU_i$ only when $SU_i$ is active. Therefore, the term $\alpha_i$ is multiplied with all economic costs in Table~\ref{my-label14}.

\begin{table}[h]
	\centering
	\caption{Different prices}
	\vspace{1mm}
	\label{my-label14}
	\begin{tabular}{|c|c|c|c|c|}
		\hline
		\begin{tabular}[c]{@{}c@{}}Actual\\ state\end{tabular} & \begin{tabular}[c]{@{}c@{}}Sensing\\ result\end{tabular} & \begin{tabular}[c]{@{}c@{}}$SU_i$ to\\ the FC\end{tabular} & \begin{tabular}[c]{@{}c@{}}$SU_i$'s \\ profit from\\ usage\end{tabular} & \begin{tabular}[c]{@{}c@{}}$SU_i$'s cost\\ for sensing\\ and transmission\end{tabular}\\ \hline
		$H_0$                                                  & $H_0$                                                                                                       & $\alpha_i t_ir_0^ia_i$                                                                                       & $\alpha_i t_ir_0^ib_i$              &  $\alpha_i\mathcal{T}$                                            \\ \hline
		$H_0$                                                  & $H_1$                                                                                                       & -                                                                                                   & -              &  $\alpha_i\mathcal{T}$                                                                       \\ \hline
		$H_1$                                                  & $H_0$                                                                                   & $\alpha_it_ir_1^ia_i$                                                                                      & $\alpha_it_ir_1^ib_i$              &  $\alpha_i\mathcal{T}$                                                           \\ \hline
		$H_1$                                                  & $H_1$                                                                                                       & -                                                                                                   & -              &  $\alpha_i\mathcal{T}$                                                                       \\ \hline
	\end{tabular}
\end{table}

\subsection{Utility function at the FC}\label{utility_FC}
The FC allocates different time durations to SUs to access the licensed spectrum only when the licensed band is sensed as idle. During a frame duration, if $SU_i$ is selected for resource allocation, then $SU_i$ can clear $R_i(P_{fa},k,L)$ bits/sec. from it's buffer, where, 
\begin{align}
R_i(P_{fa},k,L)=P(H_0)\left\{1-P_{FA}(P_{fa},k,L)\right\}r_0^i+
P(H_1)\left\{1-P_{D}(P_{fa},k,L)\right\}r_1^i 
\label{rate1}
\end{align}\noindent
If $SU_i$ pays the amount $a_i$ to the FC for each successfully transmitted bit to the FC, then the utility function at the FC becomes:
\begin{align}
U_{FC}(P_{fa},k,\boldsymbol{t},\boldsymbol{\alpha}) = \sum_{i=1}^M \alpha_iU_{FC}^i(P_{fa},k,L)t_i.
\label{payoff}
\end{align} \noindent
where, $\alpha_i$ is the activity factor of $SU_i$ as has been defined earlier in Section~\ref{sysmod} and $U_{FC}^i(P_{fa},k,L)=R_i(P_{fa},k,L)a_i$.

\subsection{Utility function for $SU_i$}\label{utility_ST_section}
Each SU's utility function consists of two parts, which we discuss in detail as follows:

\begin{itemize}
	\item $SU_i$ pays amount $a_i$ to the FC for each successful bit transmission to the FC; whereas, $SU_i$ generates revenue $b_i$ for each successful bit transmission. Corresponding utility function can be written as:
	\begin{align}
	U_{SU_i}^1 = \alpha_iR_i(P_{fa},k,L) t_i(b_i-a_i).
	\end{align}
	
	\item Each SU losses energies for conducting sensing and transmitting hard decisions, which is represented as
	\begin{align}
	U_{SU_i}^2 = \alpha_i\left(Na_s+a_t\right).
	\end{align}
\end{itemize} 

Effective utility for $SU_i$ becomes:
\begin{align}
U_{SU_i}(P_{fa},k,t_i,\alpha_i) &= U_{SU_i}^1 - U_{SU_i}^2.
\label{utility_ST}
\end{align}

\subsection{System constraints}\label{cons}
We now define constraints which are considered during the resource allocation problem.

\begin{enumerate}
	\item \textit{Positive utility at each SU:} $SU_i$ participates in the cooperative sensing if it's own utility function becomes positive. Therefore, we can write:
	\begin{align}
	U_{SU_i}(P_{fa},k,t_i,\alpha_i) \geq 0
	\label{cons1}
	\end{align}\noindent
	where, $U_{SU_i}(P_{fa},k,t_i,\alpha_i)$ has been defined in Equation~\eqref{utility_ST}.
	
	\item \textit{Buffer data constraint:} We consider data buffer at each SU. The FC allocates time to SUs according to their buffer sizes. Allocating more time to $SU_i$ than what it (i.e., $SU_i$) requires to clear the buffer, becomes a wastage. For this reason, we consider following constraint:
	\begin{align}
	\alpha_i R_i(P_{fa},k,L)t_i \leq B_i
	\label{buffer_cons}
	\end{align}\noindent
	where, $B_i$ denotes the number of bits stored in $SU_i$'s buffer and $R_i(P_{fa},k,L)$ is given in Equation~\eqref{rate1}. 
	
	\item \textit{Total time constraint:} If we consider the total frame duration as $T$, then we can write:
	\begin{align}
	\sum_{i=1}^M \alpha_i t_i \leq T^{'}(L)\label{cons2_2_4}	
	\end{align}\noindent
	where, $T^{'}(L)=T^{'}-\sum_{i=1}^M \alpha_i\tau_r^{'}$ is the effective usable time, $T^{'}=T-\sum_{k=2}^5 \tau_k$ ($\tau_{2-5}$ have been defined in Section~\ref{sysmod}).
	
	\item \textit{Interference constraint at the FC:} We consider an lower bound on the detection probability at the FC to maintain created interference on the primary network under a certain threshold. Mathematically, we can write this constraint as:
	\begin{align}
	P_{D}(P_{fa},k,L) \geq \zeta
	\label{penalty1}
	\end{align}
\end{enumerate}

\section{Optimal resource allocation problem}\label{optimResource}
 We aim to design SUs' and the FC's decision thresholds, SUs' activity factors, and active SUs' allocation time, while maximizing the utility at the FC as given in Equation~\eqref{payoff}. Corresponding optimization problem at the FC is represented as:
 \begin{subequations}
 	\begin{align}
 	P1:\hspace{5mm} \underset{P_{fa},k,\boldsymbol{t},\boldsymbol{\alpha}}{\text{maximize}}\hspace{10mm}  &U_{FC}(P_{fa},k,\boldsymbol{t},\boldsymbol{\alpha}) \\
 	\text{subject to:} \hspace{10mm} &\text{Equations}~\eqref{cons1}-\eqref{penalty1} \nonumber\\
 	& k \leq L. \label{consP1_1}
 	\end{align}
 \end{subequations}\noindent
 The constraint as given in Equation~\eqref{consP1_1}, signifies that number of active SUs should be always greater than the decision threshold at the FC. 
 
 From Equation~\eqref{utility_ST}, it can be observed that for $b_i<a_i$, $SU_i$ never gets positive utility. However, for $b_i>a_i$, $SU_i$'s utility may be positive if $t_i$ is above a lower bound, which we define as:
 \begin{align}
 R_i(P_{fa},k,L)t_i(b_i-a_i) &\geq Na_s+a_t\nonumber\\
 \Rightarrow t_i &\geq T_i^{LB}(P_{fa},k,L)\label{lowerlimit}
 \end{align}\noindent
 where, $T_i^{LB}(P_{fa},k,L)=\frac{Na_s+a_t}{R_i(P_{fa},k,L)\cdot (b_i-a_i)}$. From the buffer constraint as given in Equation~\eqref{buffer_cons}, we get an upper bound on $t_i$ as:
 \begin{align}
 t_i \leq T_i^{UB}(P_{fa},k,L)
 \label{upperlimit}
 \end{align}\noindent
 where, $T_i^{UB}(P_{fa},k,L)=\frac{B_i}{R_i(P_{fa},k,L)}$. It is to be noted that $T_i^{LB}(P_{fa},k,L)$ and $T_i^{UB}(P_{fa},k,L)$ are functions of $P_{fa}$, $k$, and $L$.
 
 From Equation~\eqref{lowerlimit} and \eqref{upperlimit}, we can reqrite optimization problem $P1$ as:
 \begin{subequations}
 	\begin{align}
 	P2:\hspace{5mm} \underset{P_{fa},k,\boldsymbol{t},\boldsymbol{\alpha}}{\text{maximize}}\hspace{5mm}  &U_{FC}(P_{fa},k,\boldsymbol{t},\boldsymbol{\alpha}) \label{objP2}\\
 	\text{subject to:} \hspace{5mm} &T_i^{LB}(P_{fa},k,L) \leq \alpha_it_i \leq T_i^{UB}(P_{fa},k,L) \label{consP2_1}\\
 	& \sum_{i=1}^M \alpha_i t_i \leq T^{'}(L)\label{consP2_2}\\
 	& P_{D}(P_{fa},k,L) \geq \zeta \label{consP2_3}\\
 	& k \leq L, \label{consP2_4}
 	\end{align}
 \end{subequations}\noindent
 It can be observed that the optimization problem $P2$ is a mixed integer programming problem as it consists both discrete variables, i.e., $\boldsymbol{\alpha}$, $k$, and real variables, i.e., $P_{fa}$, $\boldsymbol{t}$. Hence, the optimization problem is not convex anymore. We try to analyse the structure of the optimization problem for given value of $\boldsymbol{\alpha}$, which leads us to following proposition. 
 
 \begin{proposition}\label{prop1}
 	For given $\boldsymbol{\alpha}$, the optimization problem $P2$ is not convex over $P_{fa}$, $k$, and $\boldsymbol{t}$. However, if we fix values for $P_{fa}$ and $k$, then the optimization problem becomes convex over $\boldsymbol{t}$.
 \end{proposition}
 \begin{proof}
We first prove the first statement with the help of Appendix~\ref{non_quasi_concavity}. In Appendix~\ref{non_quasi_concavity}, we prove that the objective function of $P1$ as given in Equation~\eqref{payoff}, is not quasi-concave over $P_{fa}$ and $k$, which means that the objective function is not concave also for $P_{fa}$ and $k$. As the objective function is not concave in lower dimension, we can conclude that the objective function is not concave in higher dimension also, i.e., for $P_{fa}$, $k$, and $\bm{t}$.
		
The objective function and constraints of $P1$ are linear functions of $\bm{t}$. Therefore, we can conclude that for given values for $\bm{\alpha}$, $P_{fa}$ and $k$, the optimization problem $P2$ belongs to convex optimization family.
\end{proof}\noindent
 From Proposition~\ref{prop1}, we can say that $P2$ is a mixed integer and non-convex optimization problem. In Algorithm~\ref{algo1}, we show different steps for solving the optimization problem $P2$. We conduct search operation over $P_{fa}$ and $k$. After fixing $P_{fa}=\bar{P}_{fa}$ and $k=\bar{k}$, we evaluate optimal values for $\boldsymbol{t}$ and $\boldsymbol{\alpha}$, which we discuss in next section. 
 \begin{algorithm}[H]
 	\footnotesize{
 		\nl \KwData{
 			\begin{itemize}
 				\item SUs set: $\mathcal{G}$
 				\item Different parameters' values for $SU_i\in \mathcal{G}$: $a_i,b_i,B_i$
 				\item Other parameters: $N,c_s,a_s,c_t,a_t$
 				\item Grid of SU's false-alarm probabilities: $P_{fa}=\left\{\frac{i}{\mathcal{L}}\right\}_{i=1}^{\mathcal{L}}$
 			\end{itemize}.}
 		\nl \KwResult{Optimal utility at the FC.}
 		\nl \For{$a=1:1:\left|\mathcal{G}\right|$}{
 			$\bar{k}\leftarrow a$\;
 			\nl \For{$b=1:1:\mathcal{L}$}{
 				\nl			$\bar{P}_{fa}\leftarrow P_{fa}(b)$\;
 				\nl			Evaluate $\boldsymbol{t}$ and $\boldsymbol{\alpha}$ for $\bar{P}_{fa}$ and $\bar{k}$ as discussed in Section~\ref{optimsolution}, which are used for evaluating the utility function at the FC\;
 			}
 		}
 		\nl Choose the set of $P_{fa},k,\boldsymbol{t}$, and $\boldsymbol{\alpha}$, for which we get maximum utility at the FC.\		
 		\caption{Evaluation of optimal utility at the FC\label{algo1}}}
 \end{algorithm}\noindent
 For notational convenience, we denote: $R_i(\bar{P}_{fa},\bar{k},L)$ by $R_i(L)$, $T_i^{UB}(\bar{P}_{fa},\bar{k},L)$ by $T_i^{UB}(L)$, $T_i^{LB}(\bar{P}_{fa},\bar{k},L)$ by $T_i^{LB}(L)$, $P_{FA}(\bar{P}_{fa},\bar{k},L)$ by $P_{FA}(L)$, $P_{D}(\bar{P}_{fa},\bar{k},L)$ by $P_{D}(L)$, and $U_{FC}^i(\bar{P}_{fa},\bar{k},L)$ by $U_{FC}^i(L)$. 

 \section{Evaluation of $\boldsymbol{\alpha}$ and $\boldsymbol{t}$ for $\bar{P}_{fa}$ and $\bar{k}$}\label{optimsolution}
 In this section, we discuss about steps for finding out values for $\boldsymbol{\alpha}$ and $\boldsymbol{t}$ for $\bar{P}_{fa}$ and $\bar{k}$. From the optimization problem $P2$ as given in Equation~\eqref{objP2}-\eqref{consP2_4}, we get three different cases based on lower and upper limits of time durations, and effective usable time duration for a given value for $L=\left|\mathcal{G}\right|$ \footnote{$\left|x\right|$ denotes the cardinality of set $x$}:
 \begin{subequations}
 \begin{align}
 &\textbf{Case-1: }\sum_{SU_i\in \mathcal{G}}T_i^{UB}(\left|\mathcal{G}\right|) \leq T^{'}(\left|\mathcal{G}\right|) \label{case1_eq}\\ 
 &\textbf{Case-2: }\sum_{SU_i\in \mathcal{G}}T_i^{LB}(\left|\mathcal{G}\right|) \leq T^{'}(\left|\mathcal{G}\right|) \leq \sum_{SU_i\in \mathcal{G}}T_i^{UB}(\left|\mathcal{G}\right|) \label{case2_eq}\\
 &\textbf{Case-3: }\sum_{SU_i\in \mathcal{G}}T_i^{LB}(\left|\mathcal{G}\right|) > T^{'}(\left|\mathcal{G}\right|) \label{case3_eq}
 \end{align}
 \end{subequations}
 Case-1 represents that summation of upper bounds of time durations for SUs of interest is lesser equal to effective usable time duration. Case-2 means effective usable time duration is in between summation of lower bounds and upper bounds of time durations for SUs of interest. Case-3 means summation of lower bounds of time durations for SUs of interest is greater than effective usable time duration. It can be observed that for Case-1 and 2, the optimization problem $P2$ becomes feasible as all SUs of the set $\mathcal{G}$ can be selected for the resource allocation procedure with the effective available time duration $T^{'}(\left|\mathcal{G}\right|)$. However, for Case-3, the resource allocation problem becomes infeasible as all SUs of the set $\mathcal{G}$ can not be selected anymore. Therefore, there exists some cases (e.g., Case-3) where selection of SUs becomes relevant to make the resource allocation problem feasible. 
 
 We perform an initial-level selection of SUs from their lower and upper limits of time durations. We check whether $SU_i$ satisfies the condition, i.e., $T_i^{UB}(L)\geq T_i^{LB}(L)$ for $L=\left|\mathcal{G}\right|$, or not. If the condition is not satisfied, then $SU_i$ is excluded from the set $\mathcal{G}$ permanently as for any values of $L<\left|\mathcal{G}\right|$, the condition, i.e., $T_i^{LB}(L)<T_i^{UB}(L)$, is not satisfied, which can be verified from Equation~\eqref{lowerlimit} and \eqref{upperlimit}. We denote the reduced set of SUs by:
 \begin{align}
 \mathcal{G}^{'} = \left\{SU_i|T_i^{LB}(\left|\mathcal{G}\right|)<T_i^{UB}(\left|\mathcal{G}\right|),SU_i\in\mathcal{G}\right\}.\label{reducedset}
 \end{align} 
 
 We define a lower bound on the number of active SUs for resource allocation procedure as $L_{LB}$ due to interference constraint (as given in Equation~\eqref{consP2_3}), where:
  \begin{align}
  L_{LB}&=\left\{L|P_{D}(L))\geq \zeta,P_{D}(L-1)< \zeta\right\} \label{min_SUs}
  \end{align}
 Now, we check whether the condition, i.e., $\left|\mathcal{G}^{'}\right|\geq L_{LB}$, is satisfied or not to make the resource allocation problem feasible.  From Equation~\eqref{Pd}, we observe that for fixed values of $\bar{P}_{fa}$ and $\bar{k}$, detection probability at the FC is a monotonically increasing function of number of SUs participating in the detection procedure. Therefore, we get an unique value for $L_{LB}\geq \bar{k}$ from Equation~\eqref{min_SUs}. If we get $\left|\mathcal{G}^{'}\right|<L_{LB}$, then the resource allocation problem becomes infeasible which can be checked from Equation~\eqref{consP2_4} and \eqref{min_SUs}. Our further discussion in this section is for the condition, i.e., $\left|\mathcal{G}^{'}\right|\geq L_{LB}$, which makes the resource allocation problem feasible.  
 
 From our observations, we propose Algorithm~\ref{algo3}, where we show different steps to find out optimal values of $\boldsymbol{\alpha}$ and $\boldsymbol{t}$ for given values of $\bar{P}_{fa}$ and $\bar{k}$. We mention termination steps in Algorithm~\ref{algo3} wherever it is applicable. Following we discuss relevant steps of Algorithm~\ref{algo3} and refer to important proofs which guarantee the optimality of Algorithm~\ref{algo3} for identical cost values, i.e., $a_i=a,b_i=b,\forall{SU_i\in \mathcal{G}^{'}}$:
 
 \begin{enumerate}
 	\item If Case-1 satisfies for available set (i.e., $\mathcal{G}^{'}$), then step-4 is executed, else we go to step-5 of Algorithm~\ref{algo3}. In Proposition~\ref{prop222}, we discuss about the optimality of terminating Algorithm~\ref{algo3} in step-4 for identical cost values, i.e., $a_i=a,b_i=b,\forall{SU_i\in \mathcal{G}^{'}}$.	
 	
 	\item  When $\mathcal{G}^{'}$ does not satisfy Case-1, we go to step-5 of Algorithm~\ref{algo3}, where we check whether number of SUs in $\mathcal{G}^{'}$ is equal to the minimum required SUs, i.e., $L_{LB}$, as evaluated in Equation~\eqref{min_SUs}. If for $\left|\mathcal{G}^{'}\right|=L_{LB}$, in step-6 of Algorithm~\ref{algo3}, the condition, i.e., Case-2 as given in Equation~\eqref{case2_eq}, satisfies, then SUs of $\mathcal{G}^{'}$ are allocated time durations following avariant of water-filling algorithm as shown in Algorithm~\ref{algo2}\footnote{In order to make the resource allocation problem feasible, initially, all SUs are allocated their lower time bounds, which we show in step-4 of Algorithm~\ref{algo2}. The remaining time duration is distributed among SUs to maximize the FC's utility. From the expression as given in Equation~\eqref{payoff}, it can be observed that total utility at the FC is summation of all SUs' payments. Therefore, it is intuitive to give most preference to the SU which gives more payment than other SUs for same time allocation. For that reason, in step-9 of Algorithm~\ref{algo2}, the most preferred SU is tried to be allocated with it's upper time limit. This process continues until the available time duration ends up.}. However, in step-6 of Algorithm~\ref{algo3}, if $\mathcal{G}^{'}$ satisfies Case-3 as given in Equation~\eqref{case3_eq}, then the resource allocation problem becomes infeasible.  
 	
 	\item For $\left|\mathcal{G}^{'}\right|>L_{LB}$, we go to step-7 of Algorithm~\ref{algo2}. Please note that SUs of $\mathcal{G}^{'}$ may satisfy  either condition for Case-1 or Case-2. If Case-2 holds, then we follow step-9 of Algorithm~\ref{algo3}, where we show that SUs are allocated time duration following Algorithm~\ref{algo2}. However, for Case-3, we need to eliminate any SU from $\mathcal{G}^{'}$ to make the resource allocation problem feasible. In Proposition~\ref{prop3}, for identical cost values, i.e., $a_i=a,b_i=b,\forall{SU_i\in \mathcal{G}^{'}}$, we prove it is optimal to eliminate $SU_j$ (as defined in step-10 of Algorithm~\ref{algo3}) from $\mathcal{G}^{'}$ when modified $\mathcal{G}^{'}$ (constructed in step-11 of Algorithm~\ref{algo3}) satisfies condition for either Case-2 or Case-3.
 	
 	\item However, if modified $\mathcal{G}^{'}$ (constructed in step-11 of Algorithm~\ref{algo3}) holds condition for Case-1, then elimination of $SU_j$ as given in step-11 of Algorithm~\ref{algo3}, may not be optimal, which we prove in Proposition~\ref{prop3_1} for identical cost values, i.e., $a_i=a,b_i=b,\forall{SU_i\in \mathcal{G}^{'}}$. In Algorithm~\ref{algo4}, we show further steps when modified $\mathcal{G}^{'}$ (constructed in step-11 of Algorithm~\ref{algo3}) satisfies Case-1. We discuss about steps of Algorithm~\ref{algo3} in item a,b, and c, as given following:
 	
 	\begin{enumerate}
 	
	 	\item In Algorithm~\ref{algo4}, we check if there is any other set of SUs of same cardinality like modified $\mathcal{G}^{'}$ (constructed in step-11 of Algorithm~\ref{algo3}) for which we get better utility at the FC. For that, we exchange $n$ number of SUs between two sets, i.e., $\mathcal{G}_{left}$ and $\mathcal{G}_{Ex}$, where $n\leq \min{\mathcal{G}_{left},\mathcal{G}_{Ex}}$. 
 	
 \begin{figure}[t!]
 \begin{algorithm}[H]
 	\footnotesize
 	\nl \KwData{
 		\begin{itemize}
 			\item Input parameters as described in Algorithm~\ref{algo1}
 			\item Feasible SUs' set: $\mathcal{G}^{'}$
 			\item Null set to store excluded SUs: $\mathcal{G}_{Ex}=\emptyset$
 			\item Other parameter: $L_{LB}$
 		\end{itemize}}
 		\nl \KwResult{$\boldsymbol{\alpha}$ and $\boldsymbol{t}$}
 		\nl \eIf{Case-1 holds for $\mathcal{G}=\mathcal{G}^{'}$ in Equation~\eqref{case1_eq}}{
 			\nl  $\alpha_i=1,t_i=T_i^{UB}(\left|\mathcal{G}^{'}\right|),\forall{SU_i\in \mathcal{G}^{'}}$ \hspace{0.5cm}\textbf{***Termination step***} \;
 		}{
 		\nl \eIf{$\left|\mathcal{G}^{'}\right|=L_{LB}$}{
 			\nl From Equation~\eqref{case2_eq} and \eqref{case3_eq}, check which one holds for $\mathcal{G}=\mathcal{G}^{'}$ and evaluate $\boldsymbol{ \alpha}$, $\boldsymbol{t}$ using Algorithm~\ref{algo2} for Case-2. The resource allocation problem becomes infeasible if Case-3 holds.  \hspace{0.5cm}\textbf{***Termination step***}
 		}{
 		\nl \While{$\left|\mathcal{G}^{'}\right|> L_{LB}$}{	
 			\nl \If{Case-2 holds for $\mathcal{G}=\mathcal{G}^{'}$ in Equation~\eqref{case2_eq}}{	
 				\nl $\alpha_i=1,\forall{SU_i\in \mathcal{G}^{'}}$, evaluate allocated time duration vector (i.e., $\boldsymbol{t}_{\mathcal{G}^{'}}$) from Algorithm~\ref{algo2}, the FC's utility (i.e., $U_{FC}^{\mathcal{G}^{'}}$)   
 			}
 			\nl $SU_j=\arg\min_{SU_i\in \mathcal{G}^{'}} \left\{U_{FC}^i(\left|\mathcal{G}^{'}\right|)\right\}$\;
 			\nl $\mathcal{G}^{'}= \mathcal{G}^{'}\setminus SU_j$\;
 			\nl $\mathcal{G}_{Ex}=\mathcal{G}_{Ex}\cup SU_j$\; 
 			\nl \If{Case-1 holds for $\mathcal{G}=\mathcal{G}^{'}$ in Equation~\eqref{case1_eq}}{
 				\nl    Perform Algorithm~\ref{algo4} and find out $\mathcal{H}^{*}$ from step-22 of Algorithm~\ref{algo4}\;
 				
 				\nl \If{Case-1 holds for $\mathcal{G}=\mathcal{H}^{*}$ in Equation~\eqref{case1_eq}}{
 					\nl \eIf{$U_{FC}^{\mathcal{G}^{'}}>U_{FC}^{\mathcal{H}^{*}}$}{
 						\nl Consider $\boldsymbol{t}_{\mathcal{G}^{'}}$ and $\boldsymbol{\alpha}$ which are received in step-9 of Algorithm~\ref{algo3} \hspace{0.5cm}\textbf{***Termination step***}\;  
 					}{
 					\nl  $\alpha_i=1,t_i=T_i^{UB}(\left|\mathcal{H}^{*}\right|),\forall{SU_i\in \mathcal{H}^{*}}$ \hspace{0.5cm}\textbf{***Termination step***}\; 
 				}
 			}
 			\nl \If{Case-2 holds for $\mathcal{G}=\mathcal{H}^{*}$ in Equation~\eqref{case2_eq}}{
 				\nl Go to step-21 of Algorithm~\ref{algo3} considering $\mathcal{G}^{'}=\mathcal{H}^{*}$ and $\mathcal{G}_{Ex}=\mathcal{G}^{'}\setminus \mathcal{H}^{*}$\;
 			}				
 		}			
 		\nl \eIf{$\left|\mathcal{G}^{'}\right|=L_{LB}$}{
 			\nl From Equation~\eqref{case1_eq}, \eqref{case2_eq}, and \eqref{case3_eq}, check which one holds for $\mathcal{G}=\mathcal{G}^{'}$ and evaluate $\boldsymbol{ \alpha}$, $\boldsymbol{t}$ for either of Case-1 and Case-2. The resource allocation problem becomes infeasible if Case-3 holds.  \hspace{0.5cm}\textbf{***Termination step***} 
 		}{
 		\nl Go to step-7 of Algorithm~\ref{algo3}\;				
 	}										
 }
}
}
\caption{Optimal algorithm to chhose $\boldsymbol{\alpha}$ and $\boldsymbol{t}$ for $\bar{P}_{fa}$ and $\bar{k}$\label{algo3}}
\end{algorithm}
\end{figure}

	 	\item After exchanging $n$ number of SUs between two sets, i.e., $\mathcal{G}_{left}$ and $\mathcal{G}_{Ex}$, we get different sets of SUs which may satisfy any of three cases, i.e., Case-1,2, and 3. We find following interesting facts which helps in reducing complexity of Algorithm~\ref{algo4} and hence Algorithm~\ref{algo3}:
 	
		 	\begin{itemize}
		 		\item After exchanging $n$ number of SUs if sets $\mathcal{G}_1^{'}$ and $\mathcal{G}_2^{'}$ (defined in Table~\ref{my-label24}) satisfy 	Case-1, then all other sets (which are constructed after exchanging $n$ number of SUs) also satisfy Case-1. Under this condition (given in step-4 of Algorithm~\ref{algo4}), we can say that $\mathcal{G}_5^{'}$ (defined in Table~\ref{my-label24}) gives maximum utility at the FC among all possible sets of SUs constructed after exchanging $n$ number of SUs between $\mathcal{G}_{left}$ and $\mathcal{G}_{Ex}$. Corresponding proof is given in Proposition~\ref{prop42}.
 		
		 		\item After exchanging $n$ number of SUs if sets $\mathcal{G}_3^{'}$ and $\mathcal{G}_4^{'}$ (defined in Table~\ref{my-label24}) satisfy Case-2, then all other sets (which are constructed after exchanging $n$ number of SUs) also satisfy Case-2. Under this condition (given in step-7 of Algorithm~\ref{algo4}), we can say that $\mathcal{G}_6^{'}$ (defined in Table~\ref{my-label24}) gives maximum utility at the FC among all possible sets of SUs constructed after exchanging $n$ number of SUs between $\mathcal{G}_{left}$ and $\mathcal{G}_{Ex}$. Corresponding proof is given in Proposition~\ref{prop42}.
 		
		 		\item After exchanging $n$ SUs if sets $\mathcal{G}_3^{'}$ and $\mathcal{G}_4^{'}$ (defined in Table~\ref{my-label24}) satisfy Case-2 or $\mathcal{G}_4^{'}$ satisfies Case-3, then no more than $n$ number SUs are exchanged, which we prove in Proposition~\ref{prop44}.		
		 	\end{itemize}	
 	
	 	Please note that if none of conditions as given in step-4,7, and 10 of Algorithm~\ref{algo4} holds, then we need to check different combinations of SUs as shown in step-12 to 21 in Algorithm~\ref{algo4}.	
 	
	 	\item In step-22 of Algorithm~\ref{algo4}, we find the set of SUs, which is denoted by $\mathcal{H}^{*}$ for which the FC gets maximum utility among all possible sets of SUs considered in Algorithm~\ref{algo4}.
	 	
	 	\end{enumerate} 
 	
 	\item $\mathcal{H}^{*}$ may satisfy either Case-1 or Case-2, which are shown in step-15 and 19 respectively in Algorithm~\ref{algo3}. As the FC can not use the effective time duration fully for Case-1, there is a possibility of getting lower utility at the FC for $\mathcal{H}^{*}$ compared to the utility for $\mathcal{G}^{'}$ in step-9 of Algorithm~\ref{algo3}. Therefore, we put another check condition in step-16 of Algorithm~\ref{algo3}. If the received utility at the FC for $\mathcal{H}^{*}$ is more than what we received for $\mathcal{G}^{'}$ in step-9 of Algorithm~\ref{algo3} earlier, then we consider $\mathcal{H}^{*}$ to be optimal set of SUs as further exclusion of SUs from $\mathcal{H}^{*}$ will reduce the utility at the FC (proof is given in Proposition~\ref{prop222}). Otherwise we consider $\mathcal{G}^{'}$ (which is used in step-9 of Algorithm~\ref{algo3}) to be optimal set of SUs. Please note that Algorithm~\ref{algo3} is stopped when $\mathcal{H}^{*}$ satisfies Case-1.

 	\item However, if $\mathcal{H}^{*}$ satisfies Case-2, then we may get better utility at the FC by eliminating SUs further from $\mathcal{H}^{*}$ (as discussed earlier). Therefore, we return to step-7 of Algorithm~\ref{algo3} when $\mathcal{H}^{*}$ satisfies Case-2. 
\end{enumerate}

\begin{algorithm}[H]
	\footnotesize 
	\nl \KwData{
		\begin{itemize}
			\item Input parameters as described in Algorithm~\ref{algo1}
			\item SUs' sets: $\mathcal{G}^{'}$, $\mathcal{G}_{Ex}$, and $\mathcal{G}_{left}=\left\{\mathcal{G}^{'}\setminus \mathcal{G}_{Ex}\right\}$
			\item Different parameters as defined in Table~\ref{my-label24}
		\end{itemize}}
		\nl $n=1$\;
		\nl \While{$n\leq \min\left\{\left|\mathcal{G}_{left}\right|,\left|\mathcal{G}_{Ex}\right|\right\}$}{
			\nl \If{Case-1 holds for both $\mathcal{G}=\mathcal{G}_1^{'}$ and $\mathcal{G}_2^{'}$}{
				\nl $\alpha_i=1,t_i=T_i^{UB}(\left|\mathcal{G}_{left}\right|),\forall{SU_i\in \mathcal{G}_5^{'}}$; evaluate the FC's utility, i.e., $U_{FC}^{\mathcal{H}}$\;
				\nl Go to step-3 considering $n=n+1$\;
			}
			\nl \If{Case-2 holds for both $\mathcal{G}=\mathcal{G}_3^{'}$ and $\mathcal{G}_4^{'}$}{
				\nl $\alpha_i=1,\forall{SU_i\in \mathcal{G}_{6}^{'}}$, evaluate allocated time duration vector (i.e., $\boldsymbol{t}_{\mathcal{H}}$) from Algorithm~\ref{algo2} and the FC's utility (i.e., $U_{FC}^{\mathcal{H}}$)\;
				\nl Go to step-22\;
			}
			\nl \If{Case-3 holds for $\mathcal{G}=\mathcal{G}_4^{'}$}{
				\nl Go to step-22\;
			}
			\nl \If{None of condition in step-4,7, and 10 satisfies}{
				\nl Construct $\mathcal{N}=\left\{\mathcal{N}_a\right\}$ consisting of different sets of after exchanging $n$ number of SUs between $\mathcal{G}_{left}$ and $\mathcal{G}_{Ex}$, $a=1,..,{\mathcal{G}_{left} \choose n}\times {\mathcal{G}_{Ex} \choose n}$\;
				\nl $a=1$\;
				\nl \While{$a\leq {\mathcal{G}_{left} \choose n}\times {\mathcal{G}_{Ex} \choose n}$}{
					\nl \If{Case-1 holds for $\mathcal{G}=\mathcal{N}_a$ in Equation~\eqref{case1_eq}}{
						\nl $\alpha_i=1,t_i=T_i^{UB}(\left|\mathcal{G}_{left}\right|),\forall{SU_i\in \mathcal{N}_a}$\;
					}
					\nl \If{Case-2 holds for $\mathcal{G}=\mathcal{N}_a$ in Equation~\eqref{case2_eq}}{
						\nl $\alpha_i=1,\forall{SU_i\in \mathcal{N}_a}$, evaluate allocated time duration vector (i.e., $\boldsymbol{t}_{\mathcal{H}}$) from Algorithm~\ref{algo2} and the FC's utility (i.e., $U_{FC}^{\mathcal{H}}$)\;
					}
					\nl $a=a+1$\;
				}
				\nl $n=n+1$\;
			}
		}
		\nl $\mathcal{H}^{*}=\arg \max_{\mathcal{H}}\left\{U_{FC}^{\mathcal{H}}\right\}$\;
		\caption{Algorithm called in step-14 of Algorithm~\ref{algo3}\label{algo4}}
\end{algorithm}
	
\newpage	
\begin{algorithm}[H]
	\footnotesize
	\nl \KwData{
		\begin{itemize}
			\item SUs set: $\mathcal{G}^{'}$
			\item Other parameters: $r_0^i,r_1^i,a_i,b_i,\forall{SU_i\in \mathcal{G}^{'}};a_s,a_t,N,T^{'},\tau_r^{'}$
		\end{itemize}}
			\nl \KwResult{Time allocation vector $\boldsymbol{t}$ for $\mathcal{G}^{'}$}
			\nl Evaluate: $T^{'}(\left|\mathcal{G}^{'}\right|)$ and $T_i^{LB}(\left|\mathcal{G}^{'}\right|),T_i^{UB}(\left|\mathcal{G}^{'}\right|),\forall{SU_i\in \mathcal{G}^{'}}$\; 
			\nl Allocate $t_i=T_i^{LB}(\left|\mathcal{G}^{'}\right|),\forall{SU_i\in \mathcal{G}^{'}}$\;
			\nl $T_{remaining}= T^{'}(\left|\mathcal{G}^{'}\right|)-\sum_{i=1}^{\left|\mathcal{G}^{'}\right|} T_i^{LB}(\left|\mathcal{G}^{'}\right|)$\;
			\nl $C=\left|\mathcal{G}^{'}\right|$\;
			\nl \While{$T_{remaining}>0$ and $\left|\mathcal{G}^{'}\right|\leq C$}{
			\nl     $SU_j=\arg\max_{SU_i\in \mathcal{G}^{'}} \left\{U_{FC}^i(C)\right\}$\;
			\nl     $t_j = \min\left\{T_j^{UB}(C)-T_j^{LB}(C),T_{remaining}\right\}$\;
			\nl		 $T_{remaining}= (T_{remaining}-t_j)$\;		
			\nl		 $\mathcal{G}^{'}=\mathcal{G}^{'} \setminus SU_j$
		}
	\caption{Water-filling variant for time allocation when Case-2 holds\label{algo2}}
\end{algorithm}
		
\begin{table}[h!]
	\centering
	\caption{Different notations used in Algorithm~\ref{algo4}}
	\label{my-label24}
	\resizebox{\columnwidth}{!}{
		\begin{tabular}{|l|l|}
			\hline
			\multicolumn{1}{|c|}{Notation}                                                                                                                                                                                                                                                                                      & \multicolumn{1}{c|}{Description}                                                                                                                                                                                                                                                             \\ \hline
			\begin{tabular}[c]{@{}l@{}}$\boldsymbol{SU}_{A}^{B}=\left\{SU_{A}^{B}(j)\right\}_{j=1}^{\left|A\right|}$\end{tabular}                                                                                                                                   & \begin{tabular}[c]{@{}l@{}}Array of SUs of $A=\left\{\mathcal{G}_{left},\mathcal{G}_{Ex}\right\}$ arranging in descending order of\\
				$B=\left\{UB,LB,B,U\right\}$ where $UB,LB,B,$ and $U$ stand for SUs' upper \\
				time bounds, lower time bounds, buffers, and payments to the FC.\\
				As an example if we consider $A=\mathcal{G}_{left}$ and $B=UB$, then $SU_{\mathcal{G}_{left}}^{UB}(p)$\\
				has higher upper time bound than $SU_{\mathcal{G}_{left}}^{UB}(q)$ for $p<q$, $p,q=1,..,\left|\mathcal{G}_{left}\right|$.\end{tabular}                \\ \hline
			\begin{tabular}[c]{@{}l@{}}$\mathcal{G}_1^{'}$\end{tabular}                                                                                                  & \begin{tabular}[c]{@{}l@{}}Set of SUs after exchanging last $n$ number of SUs of $\boldsymbol{SU}_{\mathcal{G}_{left}}^{UB}$ with first $n$ number\\ of SUs of $\boldsymbol{SU}_{\mathcal{G}_{Ex}}^{UB}$\end{tabular}                                                                                                                                        \\ \hline
			\begin{tabular}[c]{@{}l@{}}$\mathcal{G}_2^{'}$\end{tabular} 
			& \begin{tabular}[c]{@{}l@{}}Set of SUs after exchanging first $n$ number of SUs of $\boldsymbol{SU}_{\mathcal{G}_{left}}^{UB}$ with last $n$ number\\ of SUs of $\boldsymbol{SU}_{\mathcal{G}_{Ex}}^{UB}$\end{tabular}                                                                                                   \\ \hline
			\begin{tabular}[c]{@{}l@{}}$\mathcal{G}_3^{'}$\end{tabular}                                                                                                  & \begin{tabular}[c]{@{}l@{}}Set of SUs after exchangingg last $n$ number of SUs of $\boldsymbol{SU}_{\mathcal{G}_{left}}^{LB}$  with first $n$ number\\ of SUs of $\boldsymbol{SU}_{\mathcal{G}_{Ex}}^{LB}$\end{tabular}                                                                                                        \\ \hline
			\begin{tabular}[c]{@{}l@{}}$\mathcal{G}_4^{'}$\end{tabular} 
			
			& \begin{tabular}[c]{@{}l@{}}Set of SUs after exchanging first $n$ number of SUs of $\boldsymbol{SU}_{\mathcal{G}_{left}}^{LB}$ with last $n$ number\\ of SUs of $\boldsymbol{SU}_{\mathcal{G}_{Ex}}^{LB}$\end{tabular}                                                                                                        \\ \hline
			\begin{tabular}[c]{@{}l@{}}$\mathcal{G}_5^{'}$\end{tabular}                                                                                                    & \begin{tabular}[c]{@{}l@{}}Set of SUs after exchanging last $n$ number of SUs of $\boldsymbol{SU}_{\mathcal{G}_{left}}^{B}$  with first $n$ number\\ of SUs of $\boldsymbol{SU}_{\mathcal{G}_{Ex}}^{B}$\end{tabular}                                                                                                          \\ \hline
			\begin{tabular}[c]{@{}l@{}}$\mathcal{G}_6^{'}$\end{tabular}                                                                                                    & \begin{tabular}[c]{@{}l@{}}Set of SUs after exchanging last $n$ number of SUs of $\boldsymbol{SU}_{\mathcal{G}_{left}}^{U}$ with first $n$ number \\of SUs of $\boldsymbol{SU}_{\mathcal{G}_{Ex}}^{U}$\end{tabular}                                                                                                          \\ \hline
		\end{tabular}
	}
\end{table}

\subsection{Some useful propositions}
In this section, we discuss about some important proofs based on which we conclude that our proposed algorithm, i.e., Algorithm~\ref{algo3}, is optimal.

\begin{proposition}\label{prop222}
	It is optimal to consider $\alpha_i=1,t_i=T_i^{UB}(\left|\mathcal{G}^{'}\right|),\forall{SU_i\in \mathcal{G}^{'}}$, under Case-1.
\end{proposition}

\begin{proof}
	If we consider $\alpha_i=1,\forall{SU_i\in \mathcal{G}^{'}}$, then the utility at the FC becomes:
	\begin{align}
	U_{FC}^{'} = \sum_{i=1}^{\left|\mathcal{G}^{'}\right|} U_{FC}^i(\left|\mathcal{G}^{'}\right|) T_i^{UB}(\left|\mathcal{G}^{'}\right|) \stackrel{(a)}= \sum_{i=1}^{\left|\mathcal{G}^{'}\right|}a_i B_i
	\label{utility_original}
	\end{align}\noindent
	where $(a):T_i^{UB}(\left|\mathcal{G}^{'}\right|)=B_i/R_i(\left|\mathcal{G}^{'}\right|)$, which can be verified from Equation~\eqref{upperlimit} .
	
	Now, we evaluate the utility at the FC after eliminating $SU_j\in \mathcal{G}^{'}$. After elimination, we get $L=\left|\mathcal{G}^{'}\right|-1$. From Equation~\eqref{Pf} and \eqref{Pd}, we observe that for $L=\left|\mathcal{G}^{'}\right|-1$, both false-alarm and detection probabilities reduce at the FC, which improves rate $R_l$ for $SU_l\in \left\{\mathcal{G}^{'}\setminus SU_j\right\}$ as can be verified from Equation~\eqref{rate1}. As lower and upper time bounds of $SU_l$ depend on the rate, from Equation~\eqref{lowerlimit} and \eqref{upperlimit}, we get following relations:
	\begin{subequations}
		\begin{align}
		T_{l}^{UB}(\left|\mathcal{G}^{'}\right|-1) &< T_{l}^{UB}(\left|\mathcal{G}^{'}\right|) \label{lwr_lim}\\
		T_{l}^{LB}(\left|\mathcal{G}^{'}\right|-1) &< T_{l}^{LB}(\left|\mathcal{G}^{'}\right|) \label{upp_lim}
		\end{align}
	\end{subequations}
	Moreover, we observe that effective time duration as given in Equation~\eqref{cons2_2_4}, increases, i.e.,
	\begin{align}
	T^{'}(\left|\mathcal{G}^{'}\right|-1)>T^{'}(\left|\mathcal{G}^{'}\right|)
	\label{mod_eff}
	\end{align}
	Therefore, from Equation~\eqref{lwr_lim}, \eqref{upp_lim}, and \eqref{mod_eff}, we can conclude that if the set, i.e., $\mathcal{G}^{'}$, satisfies the condition for Case-1, then the set, i.e., $\left\{\mathcal{G}^{'}\setminus SU_j\right\}$, will also satisfy the condition for Case-1. We can write the utility at the FC for $\left\{\mathcal{G}^{'}\setminus SU_j\right\}$ as:  
	\begin{align}
	U_{FC}^{''} = \sum_{i=1,i\neq j}^{\left|\mathcal{G}^{'}\right|} a_iB_i
	\label{utility_first}
	\end{align}
	From Equation~\eqref{utility_original} and \eqref{utility_first}, we observe that $U_{FC}^{''}<U_{FC}^{'}$. Hence, after we can conclude that after elimination of any SU from the set $\mathcal{G}^{'}$ reduces the utility at the FC. Therefore, it is optimal to select all SUs of the set $\mathcal{G}^{'}$, i.e, $\alpha_i,\forall{SU_i\in \mathcal{G}^{'}}$.	
\end{proof}

We get  $R_i(\left|\mathcal{G}^{'}\right|)>R_j(\left|\mathcal{G}^{'}\right|)$ when $g_{SU_i-FC}>g_{SU_j-FC}$, which can be verified from Equation~\eqref{rate1}. Therefore, for identical cost values of SUs, i.e., $a_i=a$ and $b_i=b$, $\forall{SU_i\in \mathcal{G}^{'}}$, from Equation~\eqref{lowerlimit}, we can write:
\begin{align}
T_i^{LB}(\left|\mathcal{G}^{'}\right|)<T_j^{LB}(\left|\mathcal{G}^{'}\right|)
\label{diff_relation}
\end{align}
We use the realtion as given in Equation~\eqref{diff_relation} in our further propositions. In Proposition~\ref{prop3} and \ref{prop3_1}, we prove two important steps of Algorithm~\ref{algo3} for identical cost values of SUs. We use two different sets for SUs, i.e., 
\begin{subequations}
	\begin{align}
	\mathcal{G}_1^{''}=\left\{\mathcal{G}^{'}\setminus SU_j\right\} \label{eq1111} \\
	\mathcal{G}_2^{''}=\left\{\mathcal{G}^{'}\setminus SU_l\right\} \label{eq1112}
	\end{align}
\end{subequations}
where $SU_l\in \mathcal{G}_1^{''}$ and $SU_j$ has been defined in step-10 of Algorithm~\ref{algo3} as:
\begin{align}
 SU_j=\arg\min_{SU_i\in \mathcal{G}^{'}} \left\{U_{FC}^i(\left|\mathcal{G}^{'}\right|)\right\}
 \label{elimination_SU}
\end{align}

\begin{proposition}\label{prop3}
	For identical cost values, i.e., $a_i=a$ and $b_i=b$, $\forall{SU_i\in \mathcal{G}^{'}}$, if $\mathcal{G}^{'}$ satisfies Case-2 (Case-3), then it is optimal to eliminate $SU_j$ (as given in Equation~\eqref{elimination_SU} above) when the modified set, i.e., $\mathcal{G}_1^{''}$, satisfies Case-2 (Case-2 or 3), as:
	\begin{enumerate}
		\item The FC gets more utility when $SU_j$ is eliminated than the case when $SU_l$ (defined below Equation~\eqref{eq1112}) is eliminated.
		\item The FC gets more utility for $\mathcal{G}_1^{''}$ (when satisfies condition for Case-2) than $\mathcal{G}^{'}$.
	\end{enumerate}
\end{proposition}

\begin{proof}
We first prove the first statement of this proposition. Let us first consider $\mathcal{G}^{'}$ satisfies Case-2 and prove the first statement of the proposition. If we can show that the utility at the FC for $\mathcal{G}_1^{''}$ is greater than $\mathcal{G}_2^{''}$, then the the proposition is proved. After allotting lower time durations to SUs of $\mathcal{G}_1^{''}$ and $\mathcal{G}_2^{''}$, remaining time durations can be evaluated as respectively: 
	\begin{subequations}
		\begin{align}
		T_{rest}^{'} &= T^{'}(\left|\mathcal{G}_1^{''}\right|)-\sum_{SU_i\in \mathcal{G}_1^{''}} T_{i}^{LB}(\left|\mathcal{G}_1^{''}\right|)\label{rest_time_1}\\
		T_{rest}^{''} &= T^{'}(\left|\mathcal{G}_2^{''}\right|)-\sum_{SU_i\in \mathcal{G}_2^{''}} T_{i}^{LB}(\left|\mathcal{G}_2^{''}\right|)\label{rest_time_2}
		\end{align} 
	\end{subequations}
As $\left|\mathcal{G}_1^{''}\right|=\left|\mathcal{G}_2^{''}\right|=\left|\mathcal{G}^{'}\right|-1$, we get $T^{'}(\left|\mathcal{G}_1^{''}\right|)=T^{'}(\left|\mathcal{G}_2^{''}\right|)$ and can write $T_{rest}^{'}-T_{rest}^{''}$ as
	\begin{align}
		T_{j}^{LB}(\left|\mathcal{G}^{'}\right|-1)-T_{l}^{LB}(\left|\mathcal{G}^{'}\right|-1)\label{relation}
	\end{align} 
As $R_j(\left|\mathcal{G}^{'}\right|-1)<R_l(\left|\mathcal{G}^{'}\right|-1)$, from Equation~\eqref{diff_relation}, we get:
	\begin{align}
		T_{j}^{LB}(\left|\mathcal{G}^{'}\right|-1)-T_{l}^{LB}(\left|\mathcal{G}^{'}\right|-1)&>0\nonumber\\
		\Rightarrow T_{rest}^{'}&> T_{rest}^{''}
		\label{relation}
	\end{align}
For $\mathcal{G}_1^{''}$ and $\mathcal{G}_2^{''}$, utilities at the FC are respectively:
	\begin{subequations}
		\begin{align}
			U_{FC}^{'}&= \sum_{SU_{i}\in \mathcal{G}_1^{''}} R_i(\left|\mathcal{G}^{'}\right|-1)a[T_{i}^{LB}(\left|\mathcal{G}^{'}\right|-1)+t_i^{'}]\label{payoff_1}\\
			U_{FC}^{''}&= \sum_{SU_{i}\in \mathcal{G}_2^{''}} R_i(\left|\mathcal{G}^{'}\right|-1)a[T_{i}^{LB}(\left|\mathcal{G}^{'}\right|-1)+t_i^{''}]\label{payoff_2}
		\end{align} 
	\end{subequations}		
$t_i^{'}(\geq 0)$ and $t_i^{''}(\geq 0)$ are amount of allocated time duration to $SU_{i}\in \mathcal{G}_1^{''}$ and $SU_{i}\in \mathcal{G}_2^{''}$ respectively after allocating their lower time bounds. We can evaluate:
	\begin{align}
		&U_{FC}^{'}-U_{FC}^{''} \nonumber\\
		&\stackrel{(a)}= \sum_{SU_{i}\in \mathcal{G}_1^{''}} R_i(\left|\mathcal{G}^{'}\right|-1)at_i^{'} - \sum_{SU_{i}\in \mathcal{G}_2^{''}} R_i(\left|\mathcal{G}^{'}\right|-1)at_i^{''} \nonumber\\
		&= \sum_{SU_{i}\in \mathcal{G}_1^{''}\cap \mathcal{G}_2^{''}} R_i(\left|\mathcal{G}^{'}\right|-1)a[t_i^{'} - t_i^{''}]
		+a[R_l(\left|\mathcal{G}^{'}\right|-1)t_l^{'}-R_j(\left|\mathcal{G}^{'}\right|-1)t_j^{''}]
		\label{diif_uti}
	\end{align}
where $(a):R_i(\left|\mathcal{G}^{'}\right|-1)aT_{i}^{LB}(\left|\mathcal{G}^{'}\right|-1)=a(Na_s+a_t)/(b-a), SU_i\in \mathcal{G}_1^{''}$ or $\mathcal{G}_2^{''}$, which can be checked from Equation~\eqref{lowerlimit}. 
		
As $T_{rest}^{'}>T_{rest}^{''}$, from the water filling algorithm as given in Algorithm~\ref{algo2}, we get following two conditions:
	\begin{subequations}
	 \begin{align}
		 t_i^{'}&\geq t_i^{''};SU_{i}\in \mathcal{G}_1^{''}\cap \mathcal{G}_2^{''} \label{eq2132}\\
		 t_l^{'}&\geq t_j^{''} \label{eq3212}
	 \end{align}
	\end{subequations}
From Equation~\eqref{eq2132}, we can say:
	\begin{align}
		 \sum_{SU_{i}\in \mathcal{G}_1^{''}\cap \mathcal{G}_2^{''}} R_i(\left|\mathcal{G}^{'}\right|-1)a[t_i^{'} - t_i^{''}]\geq 0
		 \label{cond3245}
	\end{align}
Whereas, from Equation~\eqref{eq3212} and the relation, i.e., $R_l(\left|\mathcal{G}^{'}\right|-1)>R_j(\left|\mathcal{G}^{'}\right|-1)$, we get:
	\begin{align}
		a[R_l(\left|\mathcal{G}^{'}\right|-1)t_l^{'}-R_j(\left|\mathcal{G}^{'}\right|-1)t_j^{''}]>0
		\label{cond4562}
	\end{align}
		
Hence from Equation~\eqref{cond3245}, \eqref{cond4562}, and \eqref{diif_uti}, we get $U_{FC}^{'}>U_{FC}^{''}$, which means exclusion of $SU_j$ is optimal when $\mathcal{G}^{'}$ satisfies condition for Case-2.
		
Now, let us consider that $\mathcal{G}^{'}$ satisfies Case-3. We have to prove that it is optimal to exclude $SU_j$ rather $SU_l$ from $\mathcal{G}^{'}$, when $\mathcal{G}_1^{''}$ satisfies Case-3. From Equation~\eqref{relation}, it can be observed that exclusion of $SU_j$ frees up more time duration compared to exclusion of $SU_l$. Therefore exclusion of $SU_j$ may help to make the resource allocation problem feasible. Now, let us consider that $\mathcal{G}_1^{''}$ satisfies Case-2, for which we can follow same steps given in this proposition to show that $\mathcal{G}_1^{''}$ gives more utility at the FC than $\mathcal{G}_2^{''}$.
		
Now we prove second statement as given in this proposition, i.e., the FC gets more utility for $\mathcal{G}_1^{''}$ (when satisfies condition for Case-2) than $\mathcal{G}^{'}$. For $\mathcal{G}_1^{''}$, we can write $\left|\mathcal{G}_1^{''}\right|=\left|\mathcal{G}^{'}\right|-1$ and moreover: $P_{FA}(\left|\mathcal{G}^{'}\right|-1)<P_{FA}(\left|\mathcal{G}^{'}\right|)$, $P_{D}(\left|\mathcal{G}^{'}\right|-1)<P_{D}(\left|\mathcal{G}^{'}\right|)$. Therefore, from Equation~\eqref{rate1}, we get:
	\begin{align}
			R_i(\left|\mathcal{G}^{'}\right|-1) &> R_i(\left|\mathcal{G}^{'}\right|) \label{124}
	\end{align}
where $SU_i\in \mathcal{G}^{'}$. We can write the difference between utilities at the FC for $\mathcal{G}_1^{''}$ and $\mathcal{G}^{'}$ as:
	\begin{align}
		 &\sum_{SU_i\in \mathcal{G}_1^{''}} R_i(\left|\mathcal{G}^{'}\right|-1)a[T_{i}^{LB}(\left|\mathcal{G}^{'}\right|-1)+t_i^{'}] - \sum_{SU_i\in \mathcal{G}^{'}} R_i(\left|\mathcal{G}^{'}\right|)a[T_{i}^{LB}(\left|\mathcal{G}^{'}\right|)+t_i] 
		 \label{123eq}
	\end{align}
where $t_i^{'} (\geq 0)$ has been defined earlier in this proposition and $t_i(\geq 0)$ is the amount of allotted time duration to $SU_i\in \mathcal{G}^{'}$ after allotting it's lower time bound. As $R_i(\left|\mathcal{G}^{'}\right|-1)\cdot T_{i}^{LB}(\left|\mathcal{G}^{'}\right|-1)=R_i(\left|\mathcal{G}^{'}\right|)\cdot T_{i}^{LB}(\left|\mathcal{G}^{'}\right|)=(Na_s+a_t)/(b-a), SU_i\in \mathcal{G}_1^{''}$ or $\mathcal{G}^{'}$, which can be checked from Equation~\eqref{lowerlimit}, we can further write Equation~\eqref{123eq} as:
	\begin{align}
		 &\sum_{SU_i\in \mathcal{G}_1^{''}\cap \mathcal{G}^{'}} a\left[R_i(\left|\mathcal{G}^{'}\right|-1)t_i^{'}-R_i(\left|\mathcal{G}^{'}\right|)t_i\right]-R_j(\left|\mathcal{G}^{'}\right|)a\left[T_j^{LB}(\left|\mathcal{G}^{'}\right|)+t_j\right]\nonumber\\
		 &=\sum_{SU_i\in \mathcal{G}_1^{''}\cap \mathcal{G}^{'}} a\left[R_i(\left|\mathcal{G}^{'}\right|-1)-R_i(\left|\mathcal{G}^{'}\right|)\right]t_i+\sum_{SU_i\in \mathcal{G}_1^{''}\cap \mathcal{G}^{'}}aR_i(\left|\mathcal{G}^{'}\right|-1)\left[t_i^{'}-t_i\right]\nonumber\\
		 &\hspace{13mm}-R_j(\left|\mathcal{G}^{'}\right|)a\left[T_j^{LB}(\left|\mathcal{G}^{'}\right|)+t_j\right]\nonumber\\
		 &\stackrel{(b)}>\sum_{SU_i\in \mathcal{G}_1^{''}\cap \mathcal{G}^{'}}aR_i(\left|\mathcal{G}^{'}\right|-1)\left[t_i^{'}-t_i\right]-R_j(\left|\mathcal{G}^{'}\right|)a\left[T_j^{LB}(\left|\mathcal{G}^{'}\right|)+t_j\right]
		 \label{tr12}
	\end{align}
where, $(b):$ from Equation~\eqref{124}. Please note that as $\mathcal{G}_1^{''}$ satisfies the condition for Case-2, we can write:
	\begin{align}
		\sum_{SU_i\in \mathcal{G}_1^{''}\cap \mathcal{G}^{'}} \left[t_i^{'}-t_i\right]>T_j^{LB}(\left|\mathcal{G}^{'}\right|)+t_j
		\label{ty23}
	\end{align}
From Equation~\eqref{ty23} and Equation~\eqref{124}, we can conclude that 
	\begin{align}
		\sum_{SU_i\in \mathcal{G}_1^{''}\cap \mathcal{G}^{'}}aR_i(\left|\mathcal{G}^{'}\right|-1)\left[t_i^{'}-t_i\right]-R_j(\left|\mathcal{G}^{'}\right|)a\left[T_j^{LB}(\left|\mathcal{G}^{'}\right|)+t_j\right]>0
	\end{align}
Hence, it is proved that due to elimination of $SU_j$ from $\mathcal{G}^{'}$ we get higher utility at the FC when $\mathcal{G}_1^{''}$ satisfies the condition for Case-2.
\end{proof}

\begin{proposition}\label{prop3_1}
	For identical cost values, i.e., $a_i=a$ and $b_i=b$, $\forall{i}$, elimination of $SU_j$ (as given in Equation~\eqref{elimination_SU} above) from $\mathcal{G}^{'}$ may not be optimal when $\mathcal{G}^{'}$ satisfies condition for Case-2/Case-3 and $\mathcal{G}_1^{''}$ satisfies condition for Case-1.
\end{proposition}

\begin{proof}
	If we can show that the utility for $\mathcal{G}_1^{''}$ is less than the utility for $\mathcal{G}_2^{''}$ (as defined in Equation~\eqref{eq1112}) for a scenario, then we can prove the claim of this proposition.
	
	From Equation~\eqref{relation}, we can observe that after allocating SUs their lower time bounds remaining time duration is more for $\mathcal{G}_2^{''}$ than $\mathcal{G}_1^{''}$. Therefore, we can say that we might get a scenario where $\mathcal{G}_2^{''}$ satisfies condition for Case-2 whereas $\mathcal{G}_1^{''}$ satisfies condition for Case-1. We prove this proposition for such a scenario. From Proposition~\ref{prop222}, we can say that $SU_i\in \mathcal{G}_1^{''}$ gets time duration $T_i^{UB}(\left|\mathcal{G}_1^{''}\right|)$. As $\left|\mathcal{G}_1^{''}\right|=\left|\mathcal{G}_2^{''}\right|=\left|\mathcal{G}^{'}\right|-1$, we can write difference between utilities for $\mathcal{G}_1^{''}$ and $\mathcal{G}_2^{''}$ as:
	\begin{align}
	\sum_{SU_i\in \mathcal{G}_1^{''}\cap \mathcal{G}_2^{''}} R_i(\left|\mathcal{G}^{'}\right|-1)a[T_{i}^{UB}(\left|\mathcal{G}^{'}\right|-1) - t_i^{'''}]
	+a[R_l(\left|\mathcal{G}^{'}\right|-1)T_{l}^{UB}(\left|\mathcal{G}^{'}\right|-1)-R_j(\left|\mathcal{G}^{'}\right|-1)t_j^{'''}]
	\label{frc}
	\end{align}
	where, $t_i^{'''}$ is allotted time duration to $SU_i\in \mathcal{G}_2^{''}$.  
	
	We consider a scenario, where for $SU_i\in \mathcal{G}_1^{''}\cap \mathcal{G}_2^{''}$, we get $T_{i}^{UB}(\left|\mathcal{G}^{'}\right|-1)=t_i^{'''}$, such that,
	\begin{align}
	\sum_{SU_i\in \mathcal{G}_1^{''}\cap \mathcal{G}_2^{''}} R_i(\left|\mathcal{G}^{'}\right|-1)a[T_{i}^{UB}(\left|\mathcal{G}^{'}\right|-1) - t_i^{'''}]=0
	\label{eq1_4}
	\end{align}
	As $\mathcal{G}_2^{''}$ satisfies Case-2, $SU_j\in \mathcal{G}_2^{''}$ is allotted: 
	\begin{align}
	t_{j}^{'''}=T^{'}(\left|\mathcal{G}^{'}\right|-1)-\sum_{SU_{i}\in \mathcal{G}_1^{''}\cap \mathcal{G}_2^{''}} T_{i}^{UB}(\left|\mathcal{G}^{'}\right|-1) \label{eq1_cond}
	\end{align}
	As $\mathcal{G}_1^{''}$ satisfies Case-1, we can say that $SU_l\in \mathcal{G}_1^{''}$ gets 
	\begin{align}
	T_{l}^{UB}(\left|\mathcal{G}^{'}\right|-1)\leq t_j^{'''} \label{eq2_cond}
	\end{align} 
	Therefore, we may get:
	\begin{align}
	  R_l(\left|\mathcal{G}^{'}\right|-1)T_{l}^{UB}(\left|\mathcal{G}^{'}\right|-1)<R_j(\left|\mathcal{G}^{'}\right|-1)t_j^{'''}
	  \label{frg}
	\end{align}  
	So, from Equation~\eqref{frg}, \eqref{eq1_4}, and \eqref{frc}, we can conclude that we can not always say 
	that the FC gets more utility for $\mathcal{G}_1^{''}$ than $\mathcal{G}_2^{''}$. Hence, exclusion of $SU_j$ (as defined below Equation~\eqref{elimination_SU}) from $\mathcal{G}^{'}$ may not be optimal when $\mathcal{G}_1^{''}$ satisfies the condition for Case-1.	
\end{proof}

Now, we prove some useful propositions which we use in Algorithm~\ref{algo4}. We use different sets of SUs, $\mathcal{G}_{Ex}, \mathcal{G}_{left}=\mathcal{G}^{'}\setminus \mathcal{G}_{Ex}$, where $\mathcal{G}_{Ex}$ is defined in step-12 of Algorithm~\ref{algo3}. Moreover, we also use sets $\mathcal{G}_{3-6}^{'}$, which are constructed with SUs from $\mathcal{G}_{Ex}$ and $\mathcal{G}_{left}$ as defined in  Table~\ref{my-label24}. 

\begin{proposition}\label{prop44}
	For identical cost values, i.e., $a_i=a,b_i=b,\forall{SU_i\in \mathcal{G}^{'}}$, it is optimal to stop exchanging more than $n$ number of SUs between $\mathcal{G}_{left}$ and $\mathcal{G}_{Ex}$ when either of following two conditions satisfy: 
	\begin{enumerate}
		\item Both $\mathcal{G}_3^{'}$ and $\mathcal{G}_4^{'}$ satisfy Case-2.
		\item $\mathcal{G}_4^{'}$ satisfies Case-3.
	\end{enumerate}  
\end{proposition}

\begin{proof}
We first prove the first statement of this proposition. It is to be noted that after exchanging $n$ number of SUs between $\mathcal{G}_{left}$ and $\mathcal{G}_{Ex}$, we can have total ${\left|\mathcal{G}_{left}\right| \choose n}\times {\left|\mathcal{G}_{Ex}\right| \choose n}$ number of sets of SUs. From Table~\ref{my-label24}, it can be verified that sum of lower time bounds of SUs of $\mathcal{G}_3^{'}$ is highest among all ${\left|\mathcal{G}_{left}\right| \choose n}\times {\left|\mathcal{G}_{Ex}\right| \choose n}$ number of sets; whereas, sum of lower bounds of SUs of $\mathcal{G}_4^{'}$ is lowest. Therefore, we can conclude that if both $\mathcal{G}_3^{'}$ and $\mathcal{G}_4^{'}$ satisfy Case-2, then remaining all sets which are constructed after exchanging $n$ number of SUs between $\mathcal{G}_{left}$ and $\mathcal{G}_{Ex}$, will satisfy Case-2. It can be observed from Algorithm~\ref{algo3} that for $SU_i\in \mathcal{G}_{left}$ and $SU_j\in \mathcal{G}_{Ex}$ we get:
	\begin{subequations}
			\begin{align}
			T_i^{LB}(\left|\mathcal{G}^{'}\right|)&<T_j^{LB}(\left|\mathcal{G}^{'}\right|) \label{eq111}\\
			U_{FC}^i(\left|\mathcal{G}^{'}\right|)&>U_{FC}^j(\left|\mathcal{G}^{'}\right|) \label{eq112}
			\end{align}
	\end{subequations}
Hence, from Equation~\eqref{eq111}, we can say that if after exchanging $n$ number of SUs between $\mathcal{G}_{left}$ and $\mathcal{G}_{Ex}$ all sets of SUs satisfy Case-2, then after exchanging $(n+l),l=1,..,\min{\left\{\left|\mathcal{G}_{left}\right|,\left|\mathcal{G}_{Ex}\right|\right\}}-n$, number of SUs all sets of SUs will either satisfy Case-2 or Case-3. As for Case-3, the resource allocation problem is not feasible, we consider sets of SUs which satisfy Case-2 after exchanging $(n+l)$ number of SUs between $\mathcal{G}_{left}$ and $\mathcal{G}_{Ex}$. It is to be noted that effective time duration, i.e., $T^{'}(\left|\mathcal{G}^{'}\right|)$, is same for either of two scenarios, i.e., for exchange of $n$ and $(n+l)$ number of SUs. Hence, from Equation~\eqref{eq112} and Algorithm~\ref{algo2}, we can conclude that the FC will get lower utility for sets of SUs which satisfy Case-2 after exchanging $(n+l)$ number of SUs. So, there is no need of exchanging more than $n$ number of SUs.
		
Now we prove the second statement of this proposition. It is to be noted that sum of lower bounds of SUs of $\mathcal{G}_4^{'}$ is lowest among all ${\left|\mathcal{G}_{left}\right| \choose n}\times {\left|\mathcal{G}_{Ex}\right| \choose n}$ number of sets which we construct after exchanging $n$ number of SUs. If $\mathcal{G}_4^{'}$ satisfies Case-3, then from Equation~\eqref{eq111}, we can conclude that we will get Case-3 only after exchanging $(n+l)$ number of SUs. Hence, we should stop exchanging any more than $n$ number of SUs.
\end{proof}

\begin{proposition}\label{prop42}
	For identical cost values, i.e., $a_i=a,b_i=b,\forall{SU_i\in \mathcal{G}^{'}}$, among all possible sets of SUs' constructed after exchanging $n$ number of SUs between $\mathcal{G}_{left}$ and $\mathcal{G}_{Ex}$, we can say:
	\begin{enumerate}
		\item The FC gets maximum utility for $\mathcal{G}_5^{'}$ when both $\mathcal{G}_1^{'}$ and $\mathcal{G}_2^{'}$ satisfy Case-1.
		\item The FC gets maximum utility for $\mathcal{G}_6^{'}$ when both $\mathcal{G}_3^{'}$ and $\mathcal{G}_4^{'}$ satisfy Case-2. 
	\end{enumerate}
\end{proposition}

\begin{proof}
We first prove the first statement of this proposition. From Table~\ref{my-label24}, it can be observed that among all possible sets of SUs' constructed after exchanging $n$ number of SUs between $\mathcal{G}_{left}$ and $\mathcal{G}_{Ex}$, summation of upper time bounds of SUs of $\mathcal{G}_2^{''}$ and $\mathcal{G}_1^{''}$ are lowest and highest respectively. Hence, we can say that if both $\mathcal{G}_1^{''}$ and $\mathcal{G}_2^{''}$ satisfy Case-1, then all possible sets of SUs', which are constructed after exchanging $n$ number of SUs between $\mathcal{G}_{left}$ and $\mathcal{G}_{Ex}$, will satisfy Case-2. From Equation~\eqref{utility_first} as given in Proposition~\ref{prop222}, we can observe that for $a_i=a,b_i=b,\forall{SU_i\in \mathcal{G}^{'}}$, the FC's utility is equal to summation of number of bits in buffers of SUs of the set which satisfies Case-2. Therefore, it can be concluded that the FC gets more utility for $\mathcal{G}_5^{'}$ compared to other all possible sets of SUs which are constructed after exchanging $n$ number of SUs between $\mathcal{G}_{left}$ and $\mathcal{G}_{Ex}$.
		
Now, we prove the second statement of this proposition. In Proposition~\ref{prop44}, we have proved that if both $\mathcal{G}_3^{'}$ and $\mathcal{G}_4^{'}$ satisfy Case-2, then all possible sets of SUs', which are constructed after exchanging $n$ number of SUs between $\mathcal{G}_{left}$ and $\mathcal{G}_{Ex}$, will satisfy Case-2. It is to be noted that among all possible sets of SUs which satisfy Case-2, $\mathcal{G}_6^{'}$ gives maximum utility at the FC.
\end{proof}

\section{Results and Discussions}\label{result}
In this section, we discuss about some results and show the efficacy of our proposed algorithm. We consider that both $g_{PT-FC}$ and $g_{SU_i-FC},\forall{SU_i\in \mathcal{G}}$, follow exponential distributions with unit mean; noise spectral density as -174 dBm/Hz., such that, effective noise power becomes $N_0=[-174+10\log_{10}(B_w)]$ dBm. Different other parameters' values which we consider in our simulation, are given in Table~\ref{my-label1}. We consider identical data type at all SUs' buffers, such that, cost values (i.e., $a_i$ and $b_i$) become identical.

\begin{table}[h]
	\centering
	\caption{Different parameters}
	\vspace{1mm}
	\label{my-label1}
	\begin{tabular}{|c|c|c|c|}
		\hline
		Parameter & Value    & Parameter      & Value                 \\ \hline
		$N$       & 40       & $P_{ST}$       & 23 dBm                \\ \hline
		$P_{PT}$  & 43 dBm       & $R_b$          & 250 Kbps              \\ \hline
		$B_w$     & 15KHz    & $f_s$          & 6 MHz. 				  \\ \hline
		$a_s$     & 0.0001   & $a_t$          &0.001   				  \\ \hline
		$b_i,\forall{SU_i\in \mathcal{G}}$       & 10       & $a_i,\forall{SU_i\in \mathcal{G}}$            &0.1   		   		  \\ \hline
	\end{tabular}
\end{table}

Please note that in Section~\ref{111111} and \ref{1121}, we take 1000 Monte-Carlo simulations, i.e., we generate $g_{SU_i-FC},\forall{SU_i\in \mathcal{G}}$ 1000 times and perform resource allocation procedure each time. Corresponding results in Section~\ref{111111} and \ref{1121} are shown taking average of all outcomes.

\subsection{Comparison with search operation}\label{111111}
In order to check optimality of our proposed algorithm, we compare with exhaustive search. We consider $P(H_0)=0.8, \gamma=-7$ dB, $B_i=1000$ bits, $g_{SU_i-FC}\sim \exp(1)$, $\forall{SU_i\in \mathcal{G}}$. We consider nine different values for $P_{fa}=[0.1,0.2,0.3,0.4,0.5,0.6,0.8,0.9]$ which are considered during search operation over $P_{fa}$. 

In Fig.~\ref{result1}(a), we plot the FC's utility for varying $\zeta$ for our proposed algorithm and exhaustive search method. We consider three different values for $M=3,5,$ and 7. It is observed that for a given value of $M$, the FC's utility decreases with $\zeta$. As the interference probability on the primary network is increased, secondary network gets less chance to access the licensed spectrum, which results reduced utility at the FC. The FC's utility improves when number of SUs is increased because more number of SUs help in improving the detection performance of the secondary network. From Fig.~\ref{result1}(a), it can be observed that our proposed algorithm gives exact result like exhaustive search, which proves that we get optimal result using our proposed algorithm.  

In Fig.~\ref{result1}(b), we plot the computational complexity (i.e., relative time required to find out the FC's optimal utility) for both the proposed algorithm and the exhaustive search method. We perform relative comparison as the outcome depends on computer specification and coding efficiency. We consider different values for $M$ and $\zeta$ and observe that convergence time increases for increasing $M$; whereas, for higher value of $\zeta$, convergence time duration reduces. It can be observed that for higher value of $M$, corresponding time of execution for exhaustive search increases significantly. However, it is almost linear for our proposed algorithm. If we consider high value for $\zeta$, then the threshold at the FC becomes high, which reduces the effective search space and hence overall execution time.
\begin{figure*}
	\centering
	\subfigure[FC utility for varying $\zeta$]{
		\label{fig_1}
		\includegraphics[width=0.7\textwidth]{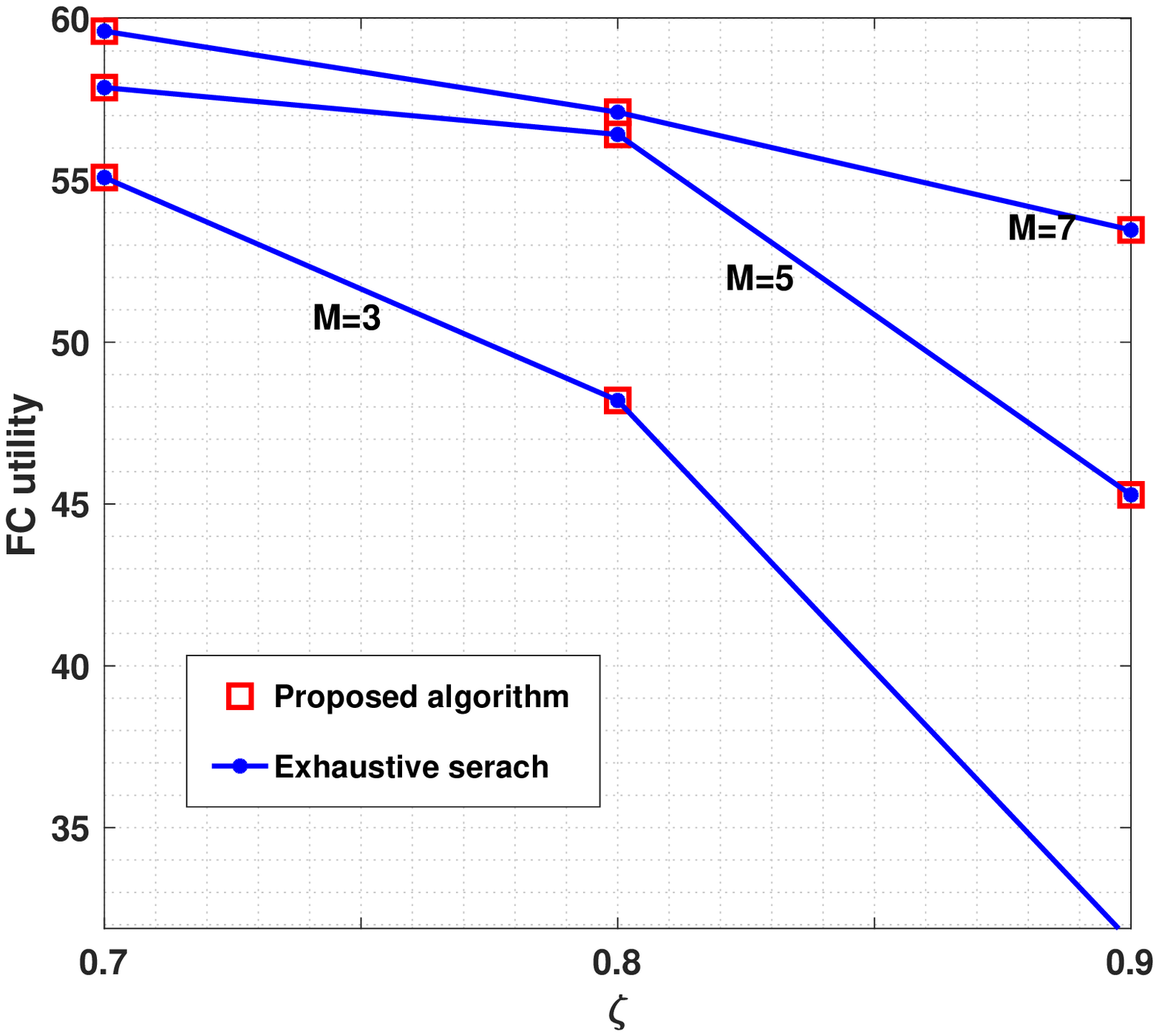}
	}
	~
	\subfigure[Time duration to get optimal result]{
		\label{fig_2}
		\includegraphics[width=0.7\textwidth]{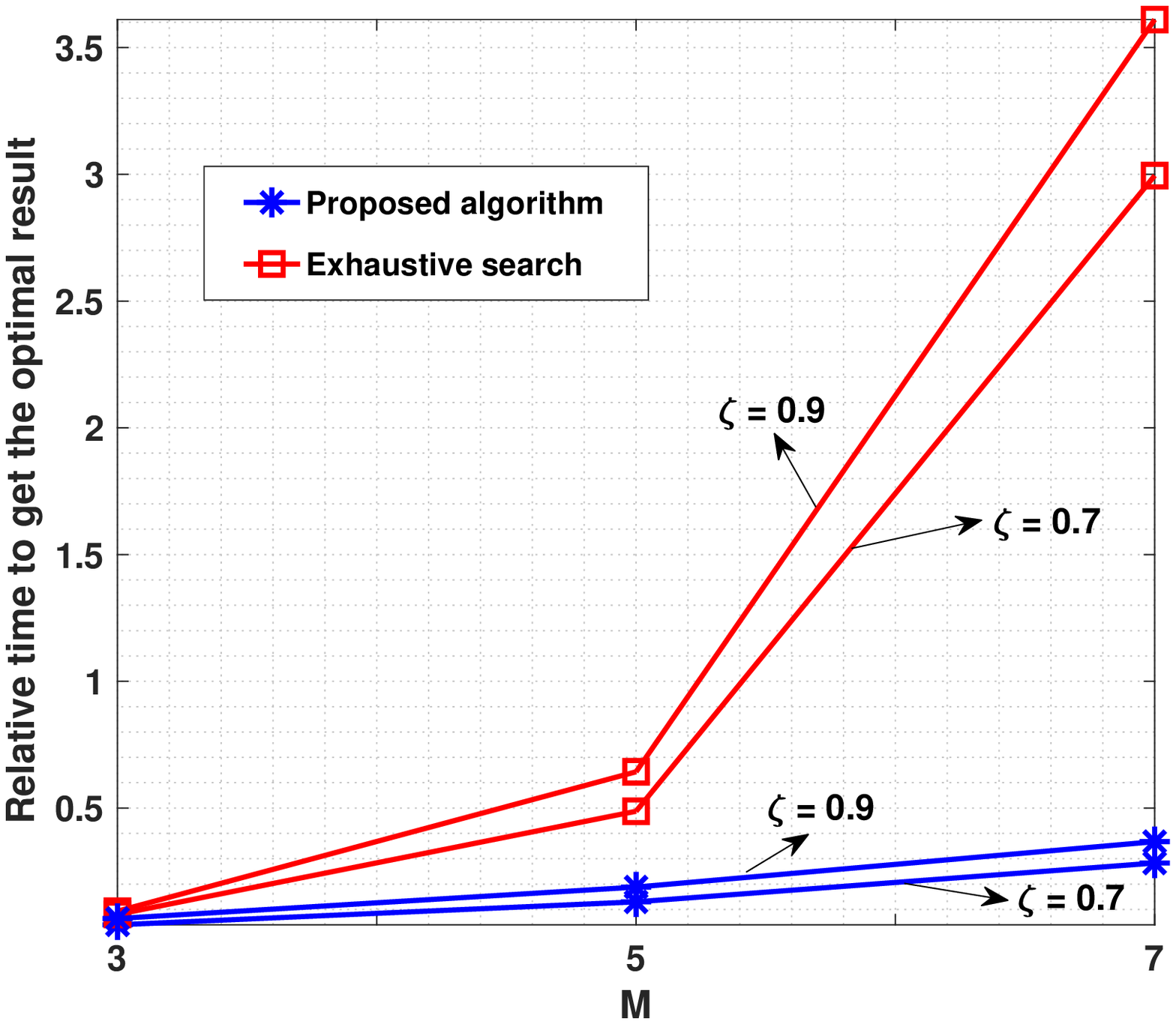}
	} 	
	\caption{Comparison with exhaustive search method}
	\label{result1}
\end{figure*}

\subsection{Advantage of joint optimization}\label{1121}
To the best of our knowledge no work has considered buffer constraint and user selection together in the context of resource allocation for opportunistic CR networks. For that reason, we can not compare our proposed algorithm with any existing framework. However, in order to understand the benefit of joint optimization (i.e., sensing thresholds, SU set for resource allocation, and allocation time durations), we compare with non-joint optimization method, which we discuss below.

\subsubsection{Non-joint optimization method}For non-joint optimization method, we consider following set of SUs:
\begin{align}
	\bar{\mathcal{G}}=\left\{SU_i|B_i(b_i-a_i)-(Na_s+a_t)\geq 0, \forall{SU_i\in \mathcal{G}}\right\}
	\label{su_set_feasible}
\end{align}  
In Equation~\eqref{su_set_feasible}, we find out SUs which get positive utility after clearing all data from their buffers. If a SU gets negative utility after clearing all data from it's buffer, then the SU will never get positive utility. Therefore, we consider $\bar{\mathcal{G}}$ in our analysis for non-joint optimization. We obtain the FC's utility solving two different optimization problems $NJ-P1$ and $NJ-P2$, which we discuss below. 
\begin{subequations}
	\begin{align}
	NJ-P1:\hspace{5mm} \underset{P_{fa},k}{\text{minimize}}\hspace{5mm}  &P_{FA}(P_{fa},k,\left|\bar{\mathcal{G}}\right|) \label{NJ-objP1}\\
	\text{subject to:} \hspace{5mm} &P_{D}(P_{fa},k,\left|\bar{\mathcal{G}}\right|) \geq \zeta \label{NJ-consP1_1}
	\end{align}
\end{subequations}
It can be observed that the optimization problem, i.e., $NJ-P1$, is purely detection oriented. For any signal detection problem, detection probability increases/decreases for fasle-alarm probability increases/decreases, which can be verified from \cite{varshney1996distributed}. Therefore, we can conclude that values for $P_{fa}=P_{fa}^{NJ-P1}$ and $k=k^{NJ-P1}$ can be solved as:
\begin{align}
  (P_{fa}^{NJ-P1},k^{NJ-P1})=\left\{(P_{fa},k)|P_{D}(P_{fa},k,\left|\bar{\mathcal{G}}\right|) = \zeta\right\}
\end{align}
We consider $\alpha_i=1,\forall{SU_i\in \bar{\mathcal{G}}}$, $P_{fa}=P_{fa}^{NJ-P1}$, and $k=k^{NJ-P1}$ to solve the following optimization problem:
\begin{subequations}
	\begin{align}
	NJ-P2:\hspace{5mm} \underset{\boldsymbol{t}}{\text{minimize}}\hspace{5mm}  &U_{FC}(P_{fa}^{NJ-P1},k^{NJ-P1},\boldsymbol{t}) \label{NJ-objP2}\\
	\text{subject to:} \hspace{5mm} &R_i(P_{fa}^{NJ-P1},k^{NJ-P1},\left|\bar{\mathcal{G}}\right|)t_i \leq B_i, \forall{SU_i\in \bar{\mathcal{G}}} \label{NJ-consP1_2}\\
	&\sum_{SU_i\in \bar{\mathcal{G}}} t_i \leq T^{'}(\left|\bar{\mathcal{G}}\right|)\label{NJ-consP2_2}
	\end{align}
\end{subequations}
where $T^{'}(\left|\bar{\mathcal{G}}\right|)=T^{'}-\sum_{i=1}^{\left|\bar{\mathcal{G}}\right|} \alpha_i\tau_r^{'}$, is effective time duration. It can be observed that we can use Algorithm~\ref{algo2} to solve $NJ-P2$ while considering $T_i^{LB}(P_{fa}^{NJ-P1},k^{NJ-P1},\left|\bar{\mathcal{G}}\right|)=0, \forall{SU_i\in \bar{\mathcal{G}}}$. Like joint optimization method, constraint on SUs' utility is not there for non-joint optimization method.

In Fig.~\ref{result2}, we plot the FC's utility received for joint and non-joint optimization methods while varying $\zeta$. We consider $M=5, P(H_0)=0.8, B_i=1000$ bits, $\forall{SU_i\in \mathcal{G}}$ and perform the comparison for two different values for $\gamma=-5$ and -7dB. It is observed that the joint method, which we have proposed, gives better result than the non-joint method. The FC's utility reduces when we increase the value for $\zeta$, which happens as the secondary network gets less opportunity to access the licensed spectrum. For increasing the sensing SNR at SUs, the FC's utility improves due to better detection performance at SUs and hence the FC. At lower sensing SNR, efficacy of joint optimization method is more prominent.
\begin{figure}[t!]
	\centering
	\includegraphics[trim=0cm 0cm 0cm 0cm,clip=true,width=12cm]{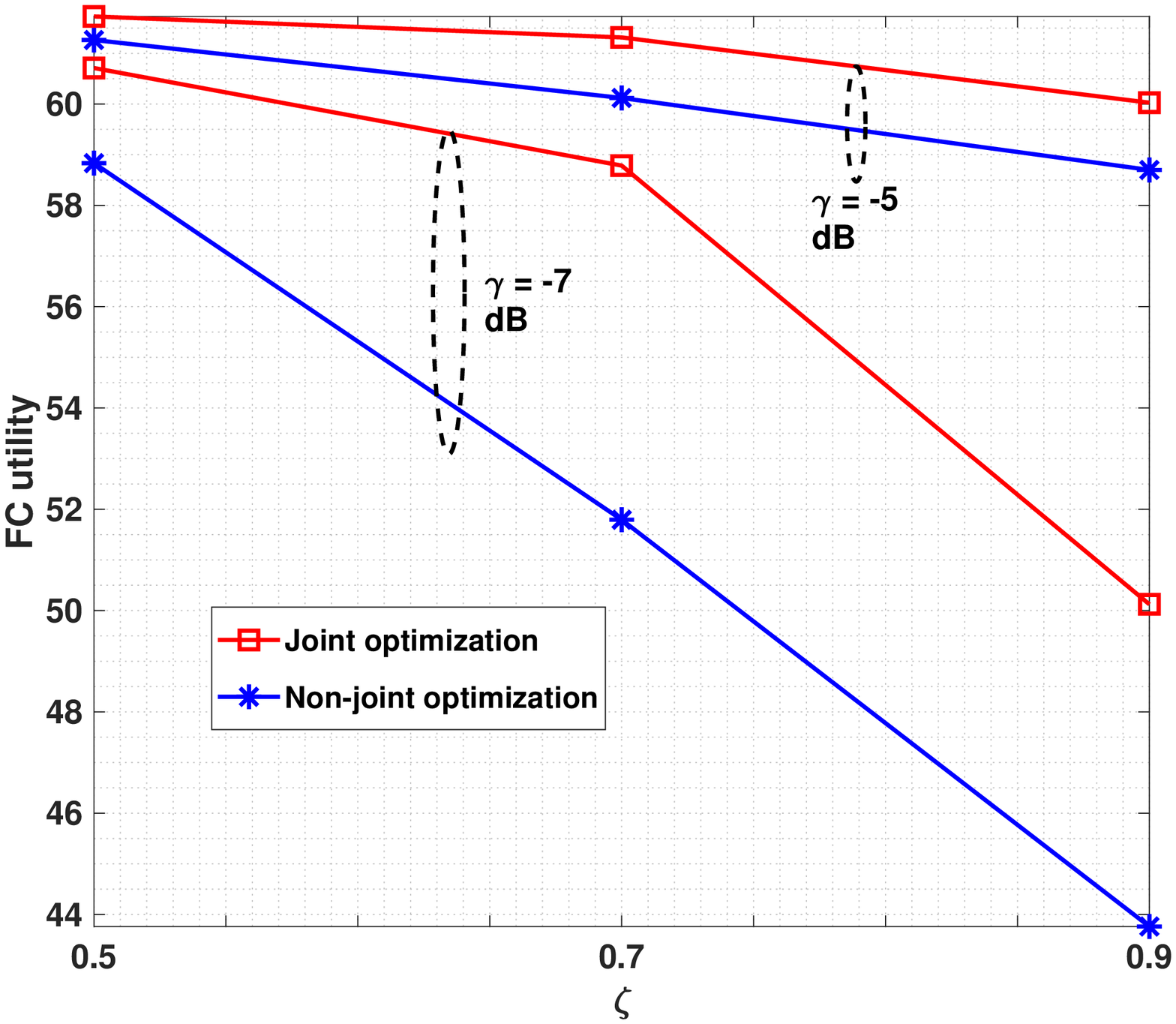}
	\DeclareGraphicsExtensions{.eps}
	\caption{FC utility vs $\zeta$}
	\label{result2}
\end{figure}

Please note that for non-joint optimization, we do not consider positive utility constraint for SUs. Therefore, some SUs get negative utility after the resource allocation. In Fig.~\ref{result6}, we consider $\gamma=-5$ dB and plot number of SUs get negative utility while varying data size in SUs' buffers, i.e., $B_i, \forall{SU_i\in \bar{\mathcal{G}}}$. We observe that as the buffer size and average value of channel gain, i.e., $g_{ST_i-FC}, \forall{SU_i\in \bar{\mathcal{G}}}$ increase, more number of SUs that get negative utility. From the expression of upper time bound as given in Equation~\eqref{upperlimit}, it can be observed that upper limit on allotted time duration increases as buffer size as well as channel gain increase. Therefore, the chance of getting time duration for a SU with lower channel gain, reduces. 
\begin{figure}[h!]
	\centering
	\includegraphics[trim=0cm 0cm 0cm 0cm,clip=true,width=12cm]{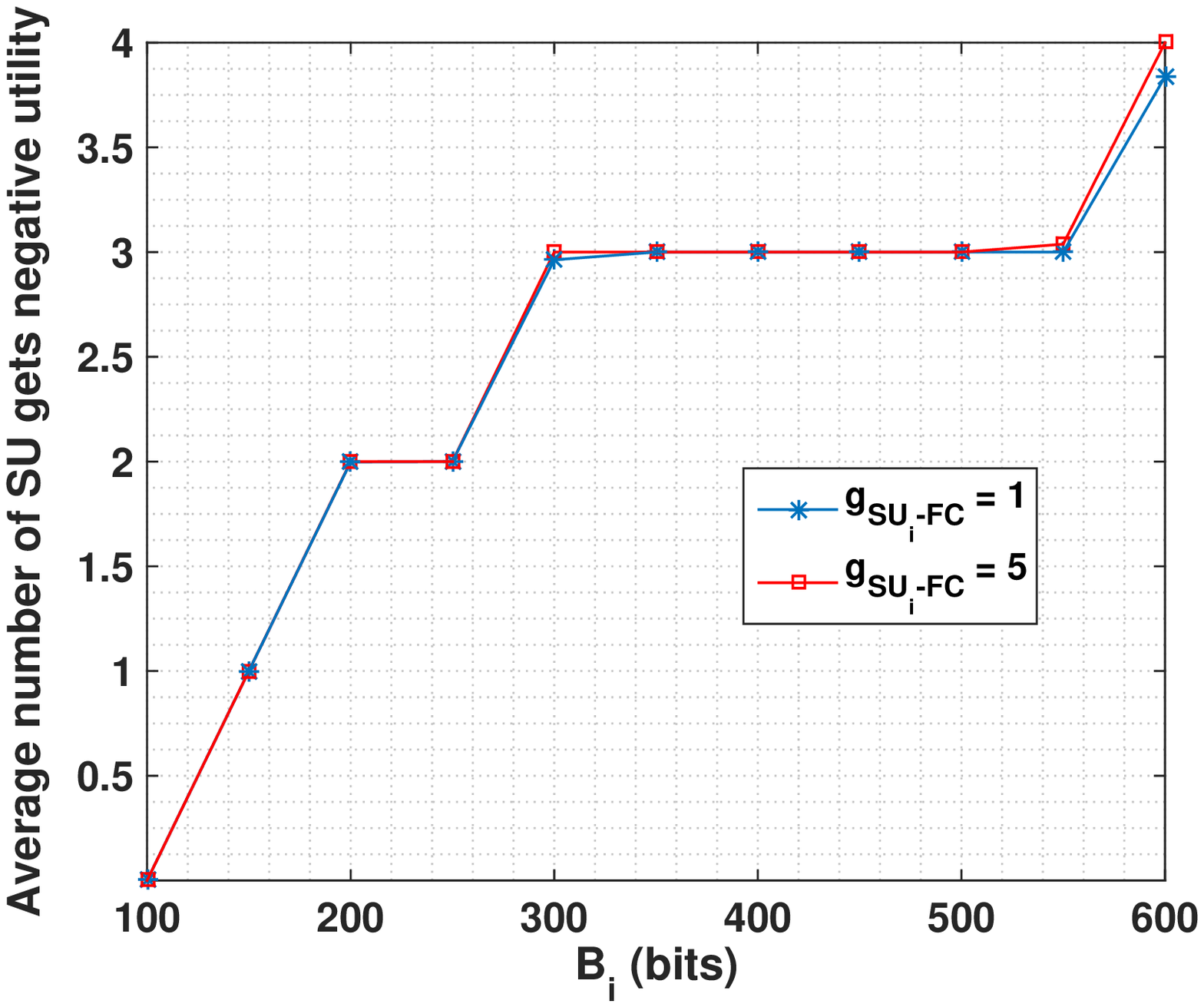}
	\DeclareGraphicsExtensions{.eps}
	\caption{Number of SU gets netaive utility for non-joint optimization for varying buffer size}
	\label{result6}
\end{figure}

\subsection{Delay measure}
In this section, we discuss about the average delay experienced by SUs in our proposed algorithm to clear their data from their buffers. Following, we discuss about the simulation set-up which we follow to generate Fig.~\ref{result11}(a) and \ref{result11}(b).

\textbf{Simulation set-up for delay measure:} In real time, the FC will perform the resource allocation for multiple frames. For that reason, we consider multi-frame simulation in order to give a clear idea about how the resource allocation procedure which we have discussed so far, may take place in real. During each frame duration, i.e., $T=1$ msec., the PU may be absent with probability $P(H_0)$ or may be present with probability $P(H_1)$. During each frame duration, we generate following channel gains:
\begin{itemize}
	\item Sensing channel gain, i.e., $g_i=\left|h_i\right|^2$ (as mentioned in Equation~\eqref{Pd}), which follows exponential distribution with mean value of $\mu_{1}$. We assume identical values for $g_i=g, \forall{SU_i\in \mathcal{G}}$.
	
	\item Channel gains between SUs and the FC, i.e., $g_{SU_i-FC}, \forall{SU_i}$ (as mentioned in Equation~\eqref{r0} and \eqref{r1}), follow exponential distribution with mean value of $\mu_{2}$.	
\end{itemize}
We assume that the FC has exact knowledge about $g_{SU_i-FC}, \forall{SU_i}$. This assumption remains valid when the control channel frequency and licensed channel frequency lie within the coherence bandwidth. Such that, estimation of channel gains between SUs and the FC in Phase-1 will correspond to channel gains between SUs and the FC in Phase-6.

At the beginning of each frame duration $SU_i$ informs to the FC about it's stored data at buffer and payment details, i.e., $a_i$ and $b_i$, $\forall{SU_i}$. In Fig.~\ref{sim}, we give a pictorial overview about how buffer data is calculated at the beginning of each frame duration for $SU_i,\forall{i}$. Traffic arrives at SUs' buffers in a random way. We assume that traffic idle time, i.e., $T_{idle}^i, \forall{SU_i}$, follows Pareto distribution, whose density function can be written as:
\begin{align}
f(T_{idle}^i) = \frac{\alpha_i \beta_i^{\alpha_i}}{(T_{idle}^i)^{\alpha_i+1}}
\end{align}
where $\alpha_i$ is the shape parameter and $\beta_i$ is the scale parameter \footnote{We assume identical parameters for SUs}. $B_{in}^i$ number of bits arrive at $SU_i$'s buffer after traffic idle time. We consider $T_b^i$ time after which $B_{in}^i$ number of bits gets accumulated at $SU_i$'s buffer. We do not use any time stamp in $B_{in}^i$ as we consider identical number of bits arrival all time. 
\begin{figure}[t!]
	\centering
	\includegraphics[trim=0cm 7cm 0cm 0cm,clip=true,width=16cm]{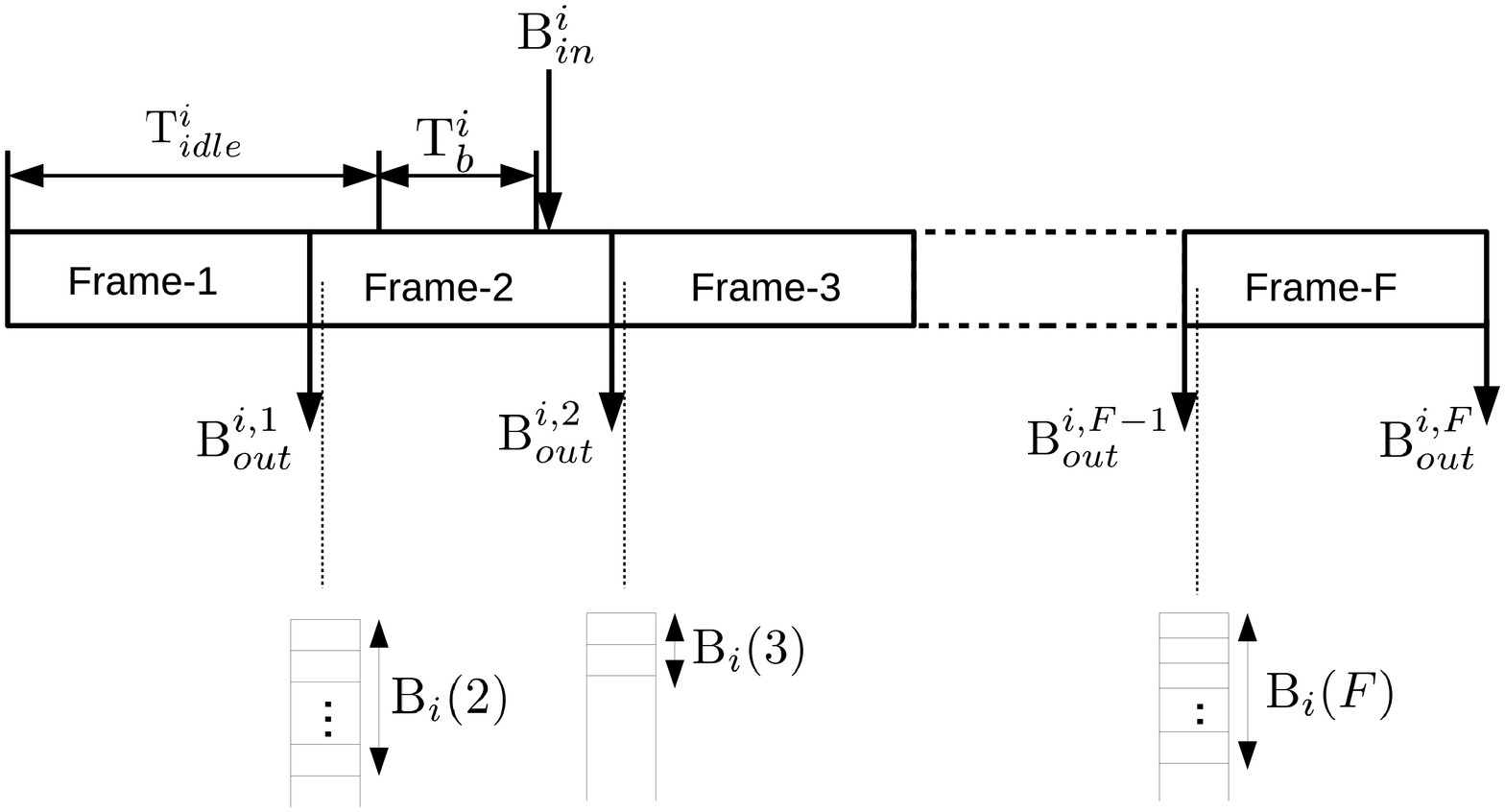}
	\DeclareGraphicsExtensions{.eps}
	\caption{Overview of buffer calculation in each frame}
	\label{sim}
\end{figure}
For example, number of bits at $SU_i$'s buffer at the beginning of third frame duration as shown in Fig.~\ref{sim}, can be calculated as:
\begin{align}
B_i(3)=B_i(2)-B_{out}^{i,2}+B_{in}^i. 
\label{buffer_calc}
\end{align}
where, $B_i(2)=B_i(1)-B_{out}^{i,1}$, is number of bits at $SU_i$'s buffer at the beginning of second frame duration and $B_i(1)=10$ bits $\forall{SU_i\in \mathcal{G}^{'}}$. From Fig.~\ref{sim}, it can be observed that during second frame duration $B_{in}^i$ number of bits arrive at $SU_i$'s buffer which we add in Equation~\eqref{buffer_calc}. We assume identical values for $\alpha_i=\alpha,\beta_i=\beta,B_{in}^i=B_{in}, T_b^i=T_b, \forall{SU_i\in \mathcal{G}}$.

For a particular frame duration, we correspond our simulation steps to different phases which we have shown in Fig.~\ref{frame} as follows:
\begin{itemize}
	\item After receiving all information from SUs, i.e., after Phase-1, as shown in Fig.~\ref{frame}, the FC finds optimal decision thresholds, optimal set for SUs from $\mathcal{G}$, and allotted time durations for selected SUs, following Algorithm~\ref{algo1}, \ref{algo3}, \ref{algo4}, and \ref{algo2}, which are discussed in Section~\ref{optimResource} and \ref{optimsolution}. Here, we assume that the FC is computationally powerful, such that, we can neglect the time for this evaluation process. 
	
	\item The FC instructs selected SUs to perform spectrum sensing in Phase-2.
	
	\item In Phase-3, selected SUs perform spectrum sensing considering the sensing channel gain, i.e., $\left|h\right|^2$, which we generate in a frame.
	
	\item After spectrum sensing, SUs send their hard decisions to the FC in Phase-4.
	
	\item The FC decides the PU to be idle when the summation of SUs' decisions is less than the FC's threshold as evaluated in Phase-1. The FC informs selected SUs about their allotted time durations in Phase-5. 
	
	\item In Phase-6, selected SUs either access (with channel gains $g_{SU_i-FC}, \forall{SU_i}$) the licensed spectrum or remain idle based on the outcome in Phase-5. Please note that we consider actual rate (i.e., bits/sec.) for SUs, e.g., if the PU follows hypothesis $H_l$ and sensing outcome is $H_m$, then rate for $SU_i$ is considered as $r_l^i;l,m\in\left\{0,1\right\}$.	
\end{itemize}

In Fig.~\ref{result11}, we plot average delays for SUs to clear their data from buffers for $M=5, \alpha=1, \beta=7, B_i(1)=10$ bits and $B_{in}^i=10$ bits, $\forall{SU_i\in \mathcal{G}}$, while varying $P(H_0)$ and $\zeta$. 

\begin{figure*}
	\centering
	\subfigure[Delay for varying $P(H_0)$]{
		\label{fig_11}
		\includegraphics[width=0.7\textwidth]{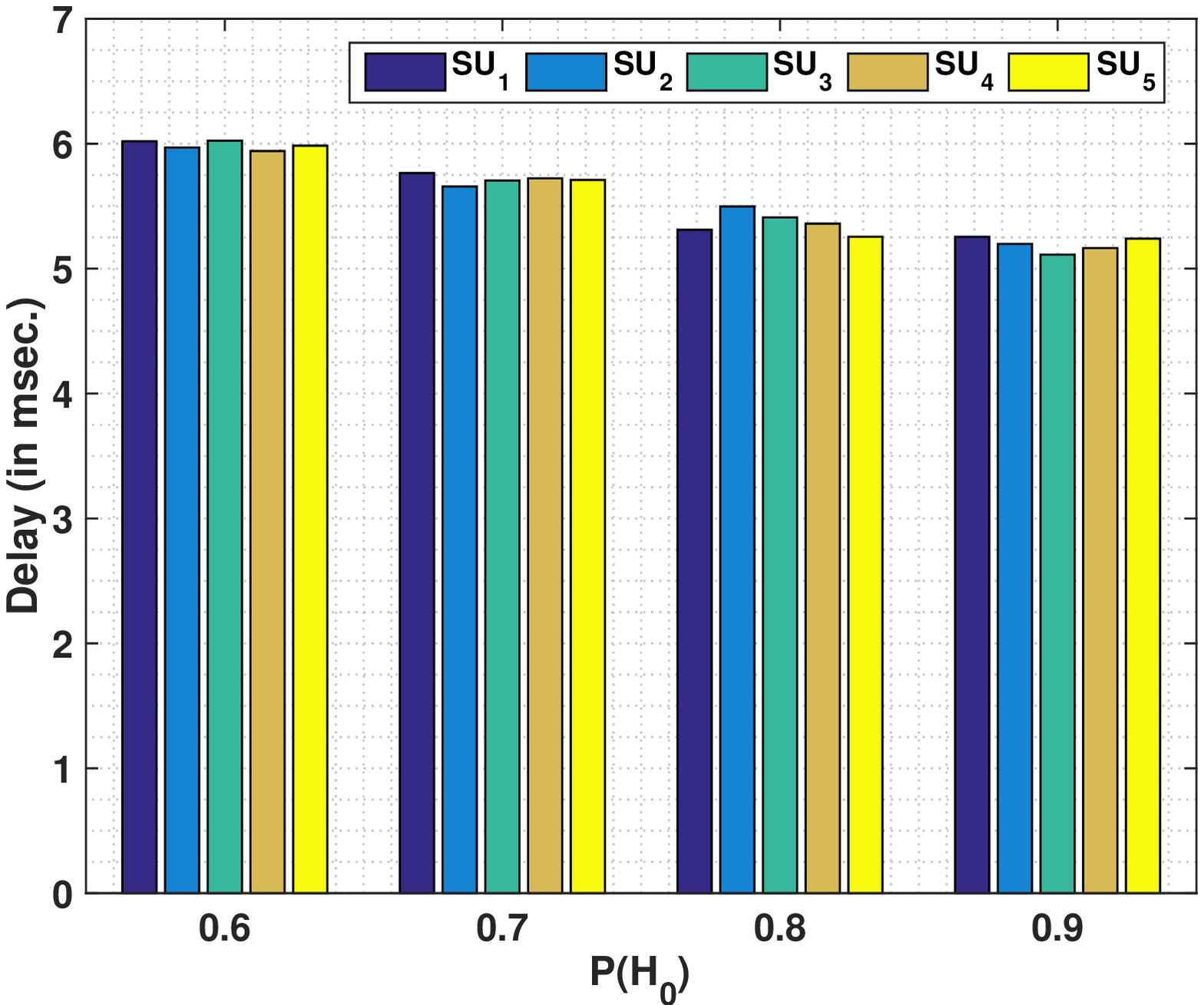}
	}
	~
	\subfigure[Delay for varying $\zeta$]{
		\label{fig_21}
		\includegraphics[width=0.7\textwidth]{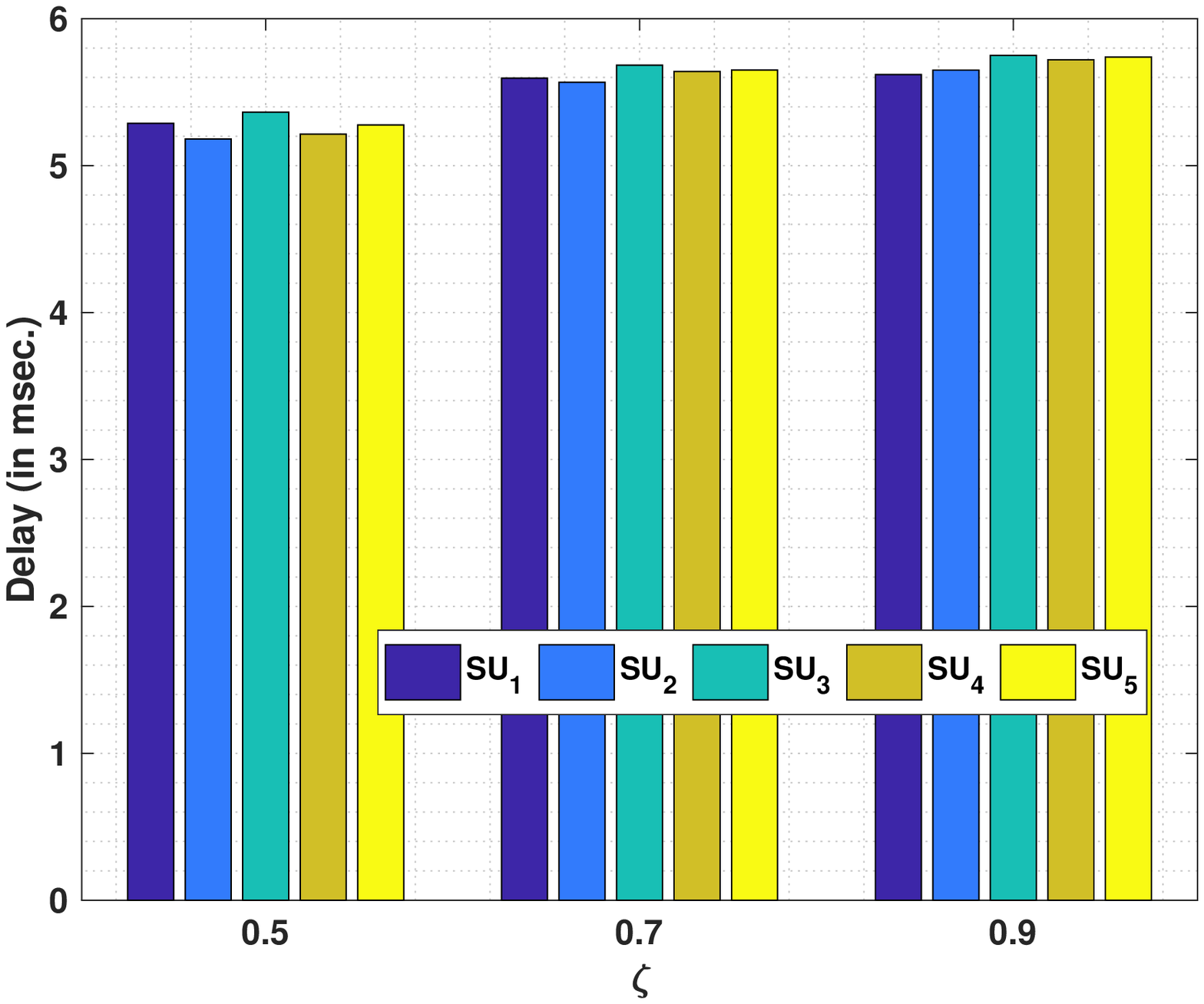}
	} 	
	\caption{Different delay values}
	\label{result11}
\end{figure*}

Fig.~\ref{result11}(a) is for $\zeta=0.7$ and $\gamma=-7$ dB, where we show that SUs take almost same time on average to clear their data from buffers for a given value of $P(H_0)$. SUs' delays reduce for increasing the value for $P(H_0)$, which is intuitive. As the PU uses the licensed band less frequently, secondary network gets more chance to access the licensed spectrum without any interference from the PU. Therefore, SUs can clear their buffer data in less time.

Fig.~\ref{result11}(b) is for $P(H_0)=0.8$ and $\gamma=-3$ dB, where we plot SUs' average delays varying $\zeta$. As $\zeta$ increases SUs take more time duration to clear their data from buffers due to less transmission opportunity to maintain stringent interference constraint. From both Fig.~\ref{result11}(a) and \ref{result11}(b), it can be observed that SUs take almost same time to clear their buffer data for a given value of $P(H_0)$ and $\zeta$, which indicates that our resource allocation algorithm is fair enough.

\subsubsection{Fairness check}
We further check the fairness of our algorithm following famous Jain's formula as given following \cite{jain1984quantitative}:
\begin{align}
 \mathcal{J}(x_1,..,x_n)=\frac{(\sum_{i=1}^N x_i)^2}{n\cdot \sum_{i=1}^n x_i^2}
 \label{Jain}
\end{align}
In our case, $x_i$ corresponds to the delay for $SU_i$. We put average delay values received from Fig.~\ref{result11}(a) and \ref{result11}(b) in Equation~\eqref{Jain}. Value of $\mathcal{J}(x_1,..,x_n)=1$ ranges from $1/n$ (worst case) to 1 (best case). We observe that for all values of $P(H_0)$ and $\zeta$, we get $\mathcal{J}(x_1,..,x_n) \approx 1$.

\section{Conclusion}\label{conclusion}
In this paper, we have considered a joint spectrum sensing and resource allocation problem for an opportunistic cooperative CR network. The FC plays dual role, i.e., as global decision maker and resource allocator. In order to get resources, SUs pay to the FC. SUs are allocated different time slots by the FC to make the SUs' transmission orthogonal over time. As the available time duration may not be sufficient for allocating time durations to all SUs, the FC may have to perform a selection procedure among SUs. Beside designing spectrum sensing thresholds, we also devise SU selection algorithm, which we prove optimal under certain condition.

\appendices

\section{$P2$ is not convex/quasi-convex over $P_{fa},k,$ and $\bm{t}$}\label{non_quasi_concavity}
 	We first prove the first statementTo prove this, we consider $\alpha_i=1,\forall{i}$. If we can show that $P2$ is not quasi-convex over $P_{fa}$ and $k$, then we can conclude that $P2$ is not convex/quasi-convex over $P_{fa},k,$and $\boldsymbol{t}$. To show this, we take help of bordered hessian matrix, which can be written for $P_{fa}$ and $k$ as follows:
 	\[
 	H=
 	\begin{bmatrix}
 	0 & \frac{\partial(U_{FC})}{\partial P_{fa}} & \frac{\partial(U_{FC})}{\partial k} \\
 	\frac{\partial(U_{FC})}{\partial P_{fa}} & \frac{\partial^2(U_{FC})}{\partial P_{fa}^2} & \frac{\partial^2(U_{FC})}{\partial P_{fa} \partial k} \\
 	\frac{\partial(U_{FC})}{\partial k} & \frac{\partial^2(U_{FC})}{\partial P_{fa} \partial k} & \frac{\partial^2(U_{FC})}{\partial k^2}
 	\end{bmatrix}
 	\] \noindent
 	We describe elements of matrix $H$ in Equation~\eqref{qceq1}-\eqref{qceq5}, where, 
 	\begin{align}
 	A&=P(H_0)\sum_{j=1}^M r_0^ia_it_i \nonumber\\
 	B&=P(H_1)\sum_{j=1}^M r_1^ia_it_i \nonumber
 	\end{align}
 	For quasi-concavity, the following rules need to be maintained:
 	\begin{subequations}
 	\begin{align}
 	\det\left[H_a\right] < 0 \label{cond1}\\
 	\det\left[H\right] > 0 \label{cond2}
 	\end{align}
 	\end{subequations}			
 	where,
 	\[
 	H_a=
 	\begin{bmatrix}
 	0 & \frac{\partial(U_{FC})}{\partial P_{fa}} \\
 	\frac{\partial(U_{FC})}{\partial P_{fa}} & \frac{\partial^2(U_{FC})}{\partial P_{fa}^2} 
 	\end{bmatrix}
 	\]

\begin{subequations}
 	\begin{align}
 		\frac{\partial(U_{FC})}{\partial P_{fa}} &= -\sum_{i=k}^M {M\choose i} \left[AP_{fa}^{i-1}(1-P_{fa})^{M-i-1}(i-MP_{fa}) + BP_{d}^{i-1}(1-P_{d})^{M-i-1}(i-MP_{d})\frac{\partial P_d}{\partial P_{fa}}\right] \label{qceq1}\\
 		\frac{\partial(U_{FC})}{\partial k} &= -A{M\choose k}P_{fa}^k(1-P_{fa})^{M-k} - B{M\choose k}P_{d}^k(1-P_{d})^{M-k} \label{qceq2} \\
 		\frac{\partial^2(U_{FC})}{\partial P_{fa}^2} &= -\sum_{i=k}^M {M\choose i} [AP_{fa}^{i-2}(1-P_{fa})^{M-i-2}\left\{(M-i-1)(i-MP_{fa})P_{fa}+i(i-P_{fa})(i-1-MP_{fa})\right\}  \nonumber\\
 		& +BP_{d}^{i-2}(1-P_{d})^{M-i-2}\left\{(M-i-1)(i-MP_{d})P_{d}+i(i-P_{d})(i-1-MP_{d})\right\}\frac{\partial P_d}{\partial P_{fa}}+ \nonumber\\
 		& BP_{d}^{i-1}(1-P_{d})^{M-i-1}(i-MP_{d})\frac{\partial^2 P_d}{\partial P_{fa}^2}] \label{qceq3}\\
 		\frac{\partial^2(U_{FC})}{\partial k^2} &= A{M\choose k} P_{fa}^k(1-P_{fa})^{M-k}\left[\frac{P_{fa}}{1-P_{fa}}\frac{1}{(K+1)(M-k-1)}-1\right] - \nonumber\\
 		& B{M\choose k} P_{d}^k(1-P_{d})^{M-k}\left[\frac{P_{d}}{1-P_{d}}\frac{1}{(K+1)(M-k-1)}-1\right] \label{qceq4}\\
 		\frac{\partial^2(U_{FC})}{\partial k \partial P_{fa}} &= A{M\choose k} P_{fa}^k(1-P_{fa})^{M-k}\left[\frac{k-MP_{fa}}{P_{fa}(1-P_{fa})}\right] + B{M\choose k} P_{d}^k(1-P_{d})^{M-k}\frac{\partial P_d}{\partial P_{fa}}\left[\frac{k-MP_{d}}{P_{d}(1-P_d)}\right] \label{qceq5}
 	\end{align}
\end{subequations} 	
 		
 \begin{figure*}
 		\centering
 		\includegraphics[trim=0cm 0cm 0cm 0cm,clip=true,width=12cm]{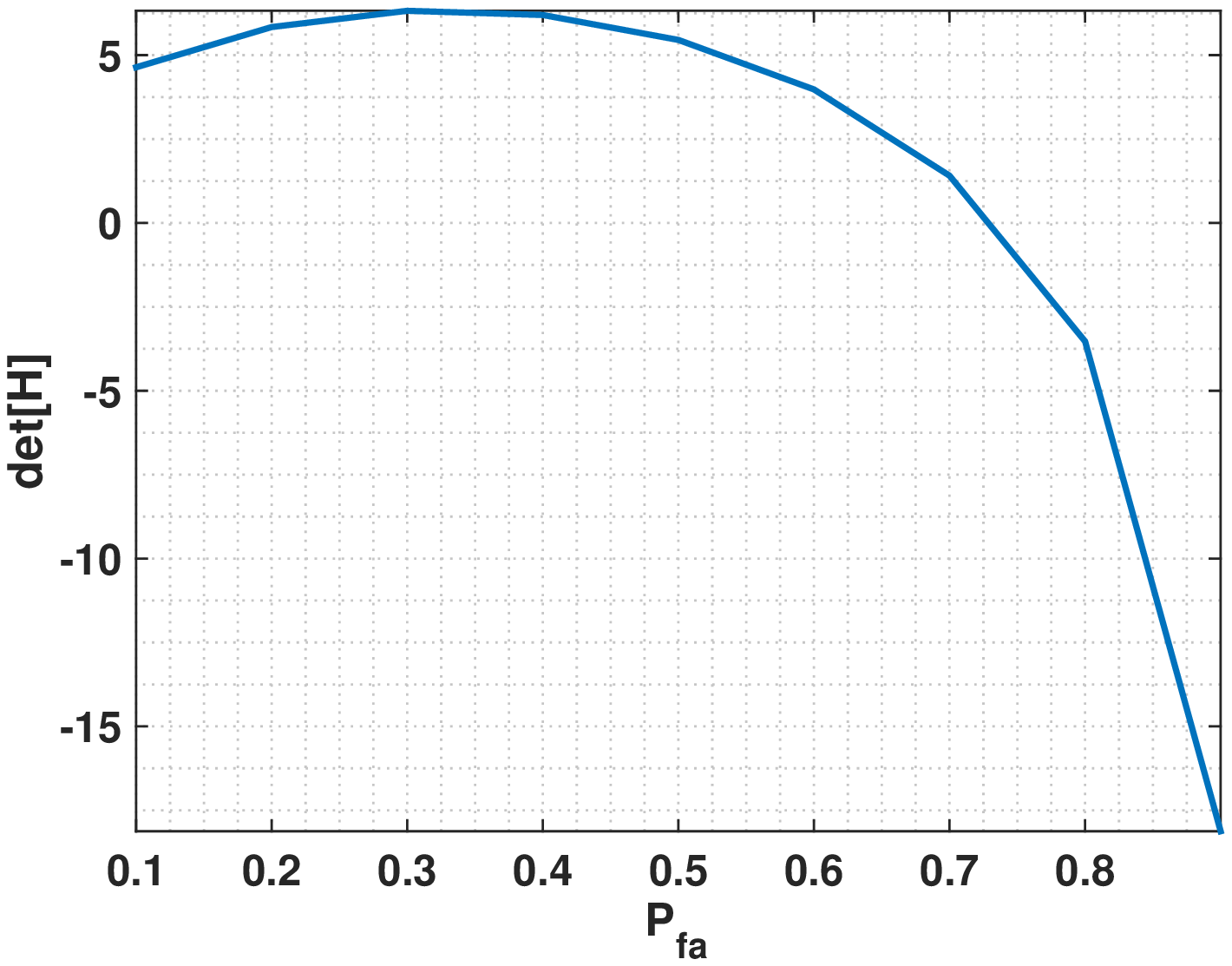}
 		\DeclareGraphicsExtensions{.eps}
 		\caption{$\det[H]$ vs $P_{fa}$}
 		\label{non-quasiconvex}
\end{figure*}
 	
 	It is observed that the condition as given in Equation~\eqref{cond1} always satisfies, however, the condition as given in Equation~\eqref{cond2}, does not satisfy always. We consider an example considering $M=k=5,P(H_0)=0.6,\gamma=-7.5$ dB, $N=40$, $\boldsymbol{r}_0=[7.4,8,8.2,0.2,9.5]$, $\boldsymbol{r}_1=[2.3,3.5,2.7,0.02,3.3]$, and $a_it_i=0.1, \forall{SU_i}$. We plot $\det[H]$ for varying $P_{fa}$ in Fig.~\ref{non-quasiconvex}, where it can be observed that $\det[H]$ is not always positive. Therefore, we can conclude that $P2$ is not quasi-convex/convex over $P_{fa}$ and $k$. As the objective function is not concave in lower dimension, we can conclude that the objective function is not concave in higher dimension also, i.e., for $P_{fa}$, $k$, and $\boldsymbol{t}$.
 	
 	The objective function and constraints of $P1$ are linear functions of $\boldsymbol{t}$. Therefore, we can conclude that for given values for $\boldsymbol{\alpha}$, $P_{fa}$ and $k$, the optimization problem $P2$ belongs to convex optimization family.

\ifCLASSOPTIONcaptionsoff
  \newpage
\fi
\bibliographystyle{IEEEtran}
\bibliography{references} 

\end{document}